\documentclass[10pt, letterpaper]{IEEEtran}

\usepackage{booktabs} 
\usepackage{amsmath,amssymb,amsfonts}
\usepackage{algorithmic}
\usepackage{graphicx}
\usepackage{textcomp}
\usepackage{xcolor}
\usepackage{multirow}
\usepackage{soul}

 \newtheorem{lemma}{Lemma}

 \newtheorem{proof}{Proof}

\begin{document}
\title{A Simple Reliability Analysis of Complex Service Function Chains (SFCs)}

	\author{Anna Engelmann and Admela Jukan}

%


\maketitle

\begin{abstract}
In the framework of Network Function Virtualization (NFV), the reliability of Service Function Chain (SFC), -- an end-to-end service is presented by a chain of virtual network functions (VNFs), is a complex function of placement, configuration and deployment requirements, both in hardware and software. Previous reliability analysis models cannot be directly applied to SFC because they do not not consider aspects of system component sharing, heterogeneity of system components and their interdependency in case of failures. In this paper, we analyze service reliability of complex SFC configurations, including serial and parallel VNF chaining, as well as their related backup protection components. Our  analysis is based on combinatorial analysis and a reduced binomial theorem, and this simple approach can be utilized to analyzing rather complex SFC configurations. We consider, for the first time, failure dependences among VNF components placed in data centers, racks, and servers. We show that our analysis can easily consider different VNF placement strategies in data center networks in arbitrary configurations and, thus, be effectively used for optimizations of the reliable SFC placement.
\end{abstract}

%
%



\section{Introduction}
Reliability of a complex system can be defined as a probability that the system will successfully complete the processing of intended service. As such, reliability analysis in system engineering is critically important to achieving system robustness in case of any system component failures. In the framework of Network Function Virtualization (NFV), the reliability of Service Function Chain (SFC) presented by a chain of virtual network functions (VNFs), is a complex function of placement, configuration and deployment requirements, both in hardware and software. SFCs suffer from processing vulnerability, i.e., failures caused by hardware or software, which can break down the entire service chain and, thus, interrupt the service. Thus, the reliability attribute of SFC is rather critical \cite{AHMED:2013,Immonen:2014}. The most common failure protection technique is based on redundant configuration and simply redirects the client's requests to backup components to maintain the service. 
\par Previous reliability analysis models cannot be directly applied to SFC because they do not not consider aspects of system component sharing, heterogeneity of system components and their interdependency in case of failures. This makes the reliability analysis of complex SFC configurations a challenge, which has not been addressed yet analytically.  Previous work consider reliability of individual system components, such as of a server, VM, link or switch, and optimize the end-to-end service reliability by solving reliable SFC mapping problem\cite{REL-ETSI, Herker:2015,Ding:2017, Hmaity:2016, Ye:2016}. For instance, in \cite{Guo:2015,Guo:2015b}, authors explored a server redundancy in data centers with server failures to reduce the cost of nonblocking multicast fat-tree data center networks (DCNs). In \cite{Soualah:2017}, a scalable and efficient algorithm for VNFs placement and chaining was proposed as a decision tree problem addressing physical link failures only. In \cite{Xu:2017}, the focus was on choosing different physical nodes to place VNFs of SFC disjointedly, i.e., one VNF per physical node, in order to avoid a single point failure and increase reliability. Papers \cite{Qu:2016,Qu:2018} studied a reliability-aware joint VNF chain placement and flow routing optimization for the case of VNF failures, whereby the number of required VNF backups was determined. These group of  problems were formulated as an Integer Linear Programming (ILP) model. 

\par Previous studies that consider failures of multiple components use assumptions such as disjointness and no interoperability between failed components. Papers \cite{Dai:2007, Dai:2010} provided a hierarchal model and reliability analysis based on Markov models, queuing theory, graph theory, and Bayesian analysis to predict service reliability by dealing with blocking, time-out, network, program and resource failure. Such reliability analysis does not consider any failure protection strategies and assumes that service is reliable, when all service components are available during service run-time. In other words, to reduce complexity the analysis requires only a multiplication of reliability in a component chain. In \cite{Nguyen:2019}, a hierarchical modeling framework is proposed of tree-based data center networks, where server, link, switches and router fail. The proposed hierarchical model consists of three layers, including reliability graphs in the top layer, a fault-tree of the subsystems and stochastic reward nets for the components in the subsystems. Also here, an approximate reliability model is provided, where all service components are disjoint and do not impact each others. Since SFC reliability is highly dependent on the reliability of individual components, whereby failure of some components can lead to failure of other components, we were inspired by idea from \cite{Engelmann:ICC18,Engelmann:2019}, which for the first time analyzed end-to-end service reliability as function of interdependent hardware (server) and software (VNFs) component failures. In this paper, we provide a complete and generalized, simple reliability analysis, that considers failures of DC, racks, servers and VMs and additionally system component sharing, heterogeneity and multitude of components as well as their interdependencies. Our  analysis is based on combinatorial analysis and a reduced binomial theorem. Surprisingly, this simple approach can be utilized to analyzing rather complex and generic SFC configurations. 
\begin{figure*}[!ht]
\centering
\includegraphics[width=1.8\columnwidth]{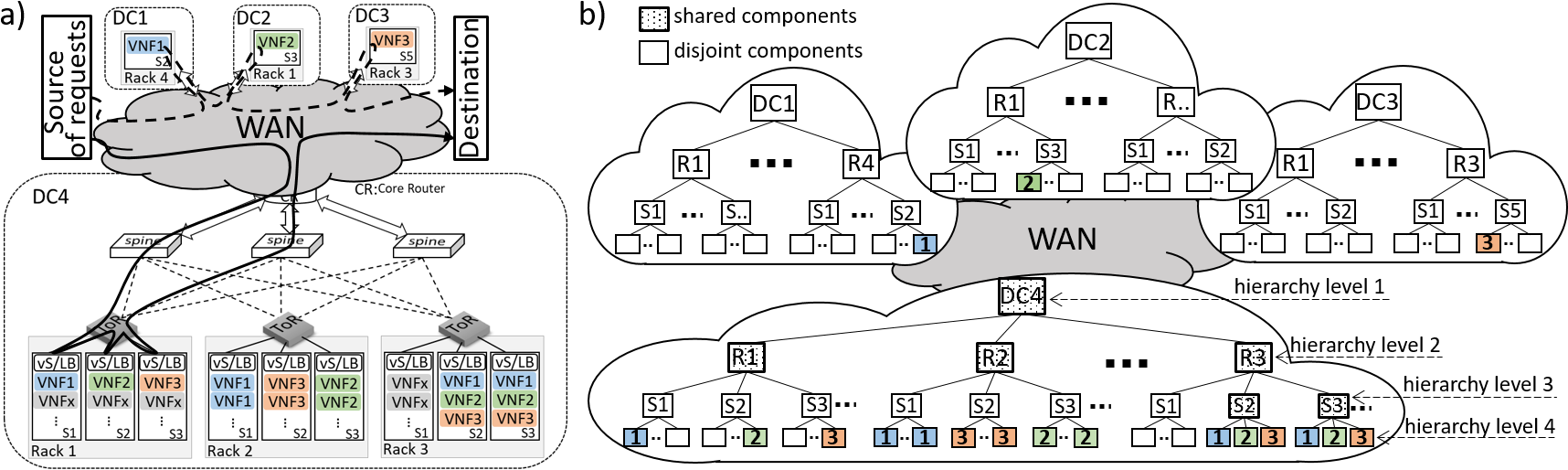}%
  \vspace{-0.3cm}
  \caption{Placement of SFCs over Inter- and Intra-DCN: a) DCN reference architecture and SFC deployment; and b) Hierarchical Component Placement Model and component interdependency.}
  \vspace{-0.5cm}
  \label{net1}
\end{figure*}
\par The rest of the paper is organized as follows. In Sec. II introduce a model of SFC. Sec. III gives a generic analysis of SFC reliability. Sec. IV shows numerical results. Sec. V concludes the paper.

\section {Preliminaries}
We introduce the basic concepts of service function chaining and VNF placement, component failures, and the related failure protection as well as SFC classification.
\subsection {Serial and Parallel SFC Placement}

An SFC is generally an ordered chain of $\Psi$ functions, i.e., VNFs. VNFs are typically allocated by Virtual Machines (VM) installed on servers in data centers (DC), whereby a number of servers are located in the same rack \cite{ETSI-Arch, CORD, GoogleDC}. We generalize the data center network architecture as a hardware fabric of switches, links, racks with multiple servers and multiple VMs per server, as illustrated in Fig. \ref{net1}a). Without loss of generality, we assume that any DCN include four types of switches: core router, Top-of-Rack (ToR), forwarding switches, e.g., spine, and Virtual Switch (vS)) to provide connections between DCN components, i.e., racks, servers, VNFs, and between DCs. Core router provides connectivity between Wide Area Network (WAN) and DCNs. ToR switches forward traffi within a rack. The forwarding switches provide connection to the core router. Each server can host multiple VMs with as many resources as required for one VNF to serve one user request. Thus, the main task of vS, e.g., programmable hypervisor switch with load balancer (LB), is to select a legitimate VM instance, i.e., VNF, for any arrived request and to provide a connection to ToR. 

Embedding VNFs of a SFC within is a challenging issue of building virtual networks with different objectives. As illustrated in Fig. \ref{net1}a), the SFCs of three VNFs, i.e., VNF1, VNF2, VNF3, can be allocated in one DC, i.e., DC4, or distributed over different DCs, i.e., DC1, DC2 and DC3 connected through WAN. Moreover, there are multiple options to allocate VNFs inside DC. For instance, in DC4, different SFCs can be allocated in the same rack or different racks, in different or same servers, but also VNFs of a certain SFC can be distributed over different servers or racks. For example, Rack1 in DC4 allocates a SFC, which is distributed over servers S1, S2 and S3, while, in Rack2, two SFCs are allocated by three servers, i.e., each server provides two VNFs of the same type. In contrast, in Rack3 (DC4), a whole SFC is allocated by one server. Each VNF of a SFC distributed over DC1, DC2 and DC3 is allocated by a separated racks and servers as well. Thus, VNF placement can be based on SFC, where VNFs of different types placed together, e.g., in the same rack, to build SFC or based on VNF type, where VNFs of the same type are placed together.

We refer to a large amount of IP packets passing the same SFC as \emph{traffic flow}. Generally, the traffic demand in an SFC can be accommodated with a single traffic flow. On the other hand, large flows can be split into parallel sub-flows and uniformly distributed over the network. Traffic can be spread over parallel SFCs using Equal Cost Multipath Protocol (ECMP), a well known technique to improve load balancing in servers and the network, i.e., to avoid server and links overload. We model parallelism in three distinctive ways: 1) Traffic Parallelism: the serial traffic flow is split into $k$ parallel \emph{sub-flows}; 2) SFC Parallelism: all VNFs of the requested SFC are replicated into at least $k$ VNF instances to process $k$ sub-flows, resulting in \emph{parallelized} SFC presented by $k$ parallel the so-called \emph{sub-SFCs}; and 3) Path Parallelism: all $k$ sub-flows are independently transmitted over $k$ link disjoint paths and processed in parallel by $k$ parallel sub-SFCs. Fig. \ref{net1}a) illustrates a parallelized SFC, where a request for SFC was split into $k=2$ sub-requests presented by solid and dashed lines. Thus, SFC is also replicated providing $k=2$ sub-SFCs, one allocated in DC4 and another in DC1, DC2 and DC3. As a result, the end-to-end service can be only successfully completed, when both sub-requests were processed by VNF1, VNF2 and VNF3, i.e., both sub-SFCs are available during service run-time.

It should be noted that when flow and SFC parallelism are deployed, we need to synchronize all VNFs of the same type. To this end, an external state repository can be utilized to store internal states of VNFs. After traffic and SFC parallelization, each sub-flow needs to pass its own sub-SFC. That can be implemented with SFC Encapsulation as proposed in \cite{RFC7665, RFC8300}.

\subsection{Hierarchical Component Placement Model}
As illustrated in Fig. \ref{net1}b), and for the purposes of reliability analysis there are $\mathcal C$ component types in the hierarchical tree with $\mathcal C$ hierarchy levels. Each hierarchy level described by a variable $c$, $1\leq c\leq\mathcal C$, also determines the component type. Considering SFC reliability with $\mathcal C=4$ components, any component of type $c=4$ is the lowest in the hierarchy and the component of type $c=1$ is the highest in the DCN hierarchy. Considering different placement strategies in Fig. \ref{net1}, it can be seen that some VNFs in an SFC are allocated in the same components, e.g., DC4 and rack 2, and, thus, share these components, i.e., rack R2 and DC4. We refer to any component which is utilized by all VNFs from the same SFC, i.e., VNFs of different type, as \textbf{shared}. In Fig. \ref{net1}b), all shared components are shown by dotted blocks. The SFCs in DC4 have a different number of shared components, i.e., SFC allocated in R1 has two shared components (R1 and DC4), similarly VNFs of SFCs placed in R2 share R2 and DC4. The SFCs allocated in S2 and S3 (R3) have three shared components each, i.e., server (S2 or S3), R3 and  DC4. 

The opposite of  shared components are the \emph{disjoint components} that separate VNFs of a certain SFC. Thus, we introduce the \textbf{level of disjointedness} $\Delta$, where $1\leq\Delta\leq \mathcal C$, which describes the number of a hierarchy levels unshared by VNFs from the same SFC, i.e., a number of component types provided to place a certain VNF type disjointedly from other VNF types. Since any VNF is placed in a separate virtual machine, i.e., disjointedly from any other VNFs, the minimal level of disjointedness is defined as $\Delta=1$. For instance, $\Delta =1$, when different VNF types are placed in the same server, where  VNFs are disjoint only, such as illustrated in Fig. \ref{net1}, i.e., SFCs placed in R3. The level of disjointedness of SFC placed in DC1, DC2 and DC3 is $\Delta=4$, where different VNF types are placed in different DCs, racks, servers and VMs and have no shared components. Generally, $\Delta$ means that $\Delta$ component types from lower hierarchy levels, i.e., from level $\mathcal C-\Delta+1$ to level $\mathcal C$, are pre-reserved for a certain VNF type, e.g., VNF1, and are disjoint from all other components of the same type $c$, which are pre-reserved for another VNF type, e.g., VNF2. We can also derive a types of shared components, i.e., hierarchy level of shared components, which varies from $\mathcal C-\Delta$ to $1$, whereby the number of shared components is $(\mathcal C-\Delta)$. 

We also introduce a \textbf{heterogeneity degree} $N_r$, where $1\leq N_r\leq \mathcal C$, which shows how many different component types are utilized to separate VNFs of the same type. When VNFs of certain type, e.g., VNFs3, are placed in different data centers (DC), racks (R), servers (S) and VMs resulting in $N_r=4$. In contrast, allocation of a certain VNF type in the same server, i.e., in different VMs only, reduces the heterogeneity degree to $N_r=1$. In Fig. \ref{net1}, the SFC placement in R3 (DC4) results in $N_r=2$ as each VNF of a certain type, e.g., VNF2, has its own VM and server, i.e., both VNFs2 are placed disjointedly in S2 and S3. However, the SFC placement in R2 (DC4) results in $N_r=1$, i.e., VNFs of the same type are allocated in the same server.

Without loss of generality, using hierarchical tree concept, for all $n=k+r$ VNFs of a certain type placed in the same fashion over DCN, where $k\geq1$ and $r\geq0$, there is a need for $n_1$, $1\leq n_1\leq n$, DCs, $n_2$, $1\leq n_2\leq n$, racks inside any DC, i.e., in total $n_1n_2$ racks, and $1\leq n_3\leq n$ servers inside any rack, i.e., in total $n_1n_2n_3$ servers, and $k+r$ VMs, whereby any server can allocate $n_4$, $1\leq n_4\leq k+r$ VNFs of the same type. To this end, we introduce a configuration set $\epsilon=\{n_1, n_2, n_3, n_4\}$ for a definition of a certain placement strategy. For instance, the configuration set can be defined as $\epsilon=\{1, 1, 1, n\}$, as $\epsilon=\{1, 1, n, 1\}$ or as $\epsilon=\{1, 1, 1, 1\}$ for SFC placed in R2, in R3 or in R1 (DC4) in Fig. \ref{net1}, respectively. 


\begin{figure*}[!t]
\centering
\includegraphics[width=1.8\columnwidth]{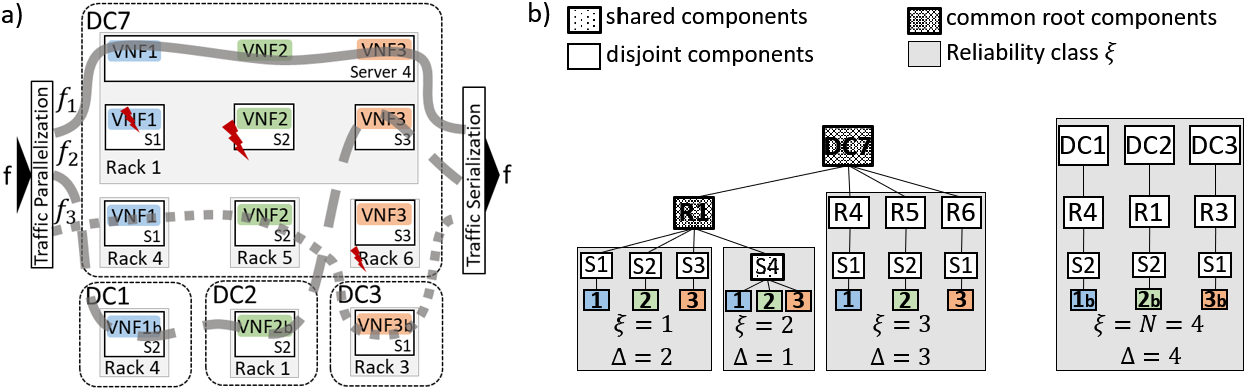}%
  \vspace{-0.3cm}
  \caption{Parallelized SFC with one backup sub-SFCs: a) Deployment of backup components; b) Hierarchical structure of complex service function chaining and reliability classification based on placement strategy.}
  \vspace{-0.6cm}
  \label{parallel}
\end{figure*}

\subsection{Component Failures}\label{failure-hierarchy}
Failure of one component type can result in failure of other component types, e.g., failed server will result in failure of all VMs installed on it. That results in a hierarchical interdependency of components, which needs to be taken into account as, \begin{itemize} 
\item{Data Center.}
The failure of data center can be caused by a failure of connectivity through WAN, failure of core router or connection between core router and forwarding switches. DC failure results in failure of all racks, servers and VMs relevant for building and maintenance of SFC.  
\item{Rack.} Rack failures can be a result of failure of forwarding or ToR switches as well as failure of links between them. Any failure of a rack will cause failure of all servers within it and all VNFs allocated in the rack's severs.
\item{Server.}
Any server can be considered as failed (unavailable) in case of failures of server hardware components, i.e., power supply, memory, etc., or due to failed link between server and ToR or in case of failure of vS, whereby all VMs allocated in the server will fail. 
\item{VM/VNF.}
Since we assume that each VM reserves as many resources as required for one VNF to serve incoming request,  failure of any VM causes failure of one VNF. We assume that VMs, i.e., VNF, on the same server are separated so that their failure has not any impact on any other VMs, i.e., VNFs.
\end{itemize}

\par To consider the interdependency of failures, let us use the model DCN as a hierarchical tree, as shown in Fig. \ref{net1}b). The root node is a DC, which belongs to the highest hierarchy level. The VNFs represent a leaf nodes of the lowest hierarchy level. The links between any nodes define interdependency between components. In the presented architecture, the components of the same type, i.e., of the same hierarchy level, are independent from each other. The components of lower hierarchy level do not impact connected components of higher hierarchy level. For example, when we consider SFC allocated in DC4, a failure of VNF1 will not affect the availability of VNF2 and VNF3 placed in R1 and servers S2 and S3, respectively (Fig. \ref{net1}b)). The availability of S1, Rack1 or DC4 will be not affected by VNF1 failure as well. In contrast, the components of higher hierarchy level impact connected components of lower hierarchy level, e.g., a failure of R1 (level 2) will result in failures of all servers (level 3) inside this rack and all VMs (level 4), i.e., failure of all VNFs (VNF1, VNF2 and VNF3). Additionally, for component failures, the failure of a whole SFCs needs to considered. As shown in Fig. \ref{net1}b), all VNFs from any SFC in DC4 have some shared components, i.e., SFC allocated in R1 has two shared components (R1 and DC4), VNFs of SFCs placed in R2 share R2 and DC4 and the SFCs allocated in S2 and S3 (R3) have shared server (S2 or S3), R3 and  DC4. As a result, a failure of any shared component will result not only in failure of components from lower hierarchy level but also in failure of an entire sub-SFC, i.e., $\Psi$ VNFs of different types.

\subsection{Failure Protection}

Without any additional resources, any SFC is characterized by a certain level of reliability, which depends on the reliability of the underlying DCN components, i.e., DCs, Racks, Servers and VMs reliability. To additionally protect the active VNFs, any SFC can be enhanced by $r$ backup VNFs resulting in $n=k+r$ sub-SFCs. The backup VNFs can then replace any failed active VNF of the same type. Similar to active VNFs, also backup VNFs can fail during service run-time. As shown in \cite{REL-ETSI, mon-ETSI}, to provide highly reliable communication, there are multiple different strategies for active monitoring of NFV, fast VNF failure detection and effective deployment of active and backup VNFs. We do not explicitly consider backup deployment strategy. Instead, we assume any service as successful, if at least $k\geq1$ out of $n$ VNFs of each type are available during service run-time, while at most $r$ VNFs can fail without causing service interruption. Depending on SFC placement strategy, a failure of a different component types will lead to a different amount of failed VNFs. 

Fig. \ref{parallel} illustrates an example of the backup protection of parallelized SFC, where a traffic flow \textbf{f} is split into $k=3$ sub-flows $f_1$, $f_2$ and $f_3$, SFC of three VNFs (VNF1-VNF2-VNF3) is replicated into $k=3$ active sub-SFCs and one backup sub-SFC (VNF1b-VNF2b-VNF3b). The sub-SFCs for sub-flows $f_1$, $f_2$ and $f_3$ are allocated in DC7. The first and second sub-SFCs are allocated in Rack1. All VNFs of the first sub-SFC are placed in the same server 4. In contrast, each VNF of the second sub-SFC for $f_2$ are allocated in different servers, i.e., S1, S2 and S3. The third sub-SFC for $f_3$ is distributed over three different racks in DC7 so, that each VNF has its own rack and server, i.e., VNF1, VNF2 and VNF3 are placed in S1 (Rack 4), S2 (Rack 5) and S3 (Rack 6), respectively. As a result, the parallel sub-flows can be sent over parallel disjoint paths. Each backup VNF of the backup sub-SFC is placed in individual DC. In DC7, VNF1 and VNF2 of active sub-SFC for $f_2$ fail due to VM (VNF1) and server S2 failure in Rack 1. As a result, the sub-flow $f_2$ is redirected toward DC1 and DC2 to be processed by backup VNFs, i.e., VNF1b and VNF2b. The active VNF3 of sub-flow $f_3$ is also failed due to failure of Rack 6 in DC7. Consequently, $f_3$ is redirected to the backup VNF3b placed in DC3. Although there are two hardware and one VM failures resulting in lost of three active VNFs of different types, the end-to-end service is successful as all three sub-flows ($f_1$, $f_2$ and $f_3$) arrive at destination and are serialized into original flow \textbf{f}.

\subsection{SFC Reliability Classes}
In this paper, we assume that any component type $c$ utilized to allocate $\Psi$ VNFs of a certain SFC have an equal component dependent reliability $p_c$, i.e., $p_c^1=...=p_c^\psi=...=p_c^\Psi=p_c$, where $1\leq c\leq\mathcal C$ and $\psi$, $1\leq \psi\leq\Psi$, shows the VNF type allocated by component $c$. In this case, the SFC build a reliability class $\xi$, whereby the component dependent reliability of different component types do not need to be the same, i.e., $p_{c-1}\neq p_c\neq p_{c+1}$. 
Generally, there can be at most $N=n=k+r$ reliability classes, i.e., $1\leq \xi\leq N$, where each sub-SFC has components with reliability different from component reliability utilized by any other SFCs. Due to the fact that, for a successful service completion, we need any $k$ out of $n$ available sub-SFC, we do not need to differentiate between active and backup SFCs, whereby both SFC types can belong to the same or different reliability classes. 
As a result, the component dependent reliability of each component type $c$ is noted as $p_{c_{\xi}}$, where $\xi$ shows the reliability class of any component type $c$ utilized to allocate VNFs from this reliability class. Depending on placement strategy and the reliability class $\xi$, there can be $n_{c_{\xi}}$ components of type $c$ inside a component of type $c+1$, whereby the configuration set is noted as $\epsilon_\xi=\{n_{1_\xi}, n_{2_\xi}, n_{3_\xi}, n_{4_\xi}\}$. When all components of a certain type $c$ used for allocation of all sub-SFCs have the same component dependent reliability, i.e., $p_{c_1}=...=p_{c_\xi}=p_{c_{\xi+1}}=...=p_{c_N}$, i.e., all components from the same hierarchy level have the same component dependent reliability independent from sub-SFC allocated, these sub-SFCs can still belong to different classes, at most $N=n$, or build one reliability class. The first option is possible, if any sub-SFC is placed over DCNs following different placement strategies, whereby each sub-SFC $\xi$ is, then, characterized by its own characterization sets of a configuration set $\epsilon_\xi$, level of disjointedness $\Delta_\xi$ and the heterogeneity degree $N_{r_\xi}$, which is different from all other characterization sets of other $n-1$ sub-SFCs.
On the other hand, if there are $n'$, $1\leq n'\leq n$, sub-SFCs, which are allocated by components with the same component reliability, i.e., $p_{c_\xi}=p{c_\xi+1}$, $p_{{c+1}_\xi}=p{{c+1}_\xi+1}$, etc., and all $n'$ sub-SFCs are placed over DCNs following the same placement strategy, i.e., have equal configuration sets $\epsilon_\xi=\epsilon_{\xi+1}$ and level of disjointedness $\Delta_\xi=\Delta_{\xi+1}$ resulting in the same heterogeneity degree $N_{r_\xi}=N_{r_{\xi+1}}$, then all $n'\equiv n_\xi$ SFCs belong to the same reliability class $\xi$. 
For instance, in Fig. \ref{net1}, the both SFCs placed in R2 or R3 in DC4 build one SFC reliability class, respectively.
The two sub-SFCs in Rack3 have $\epsilon_\xi=\epsilon_{\xi+1}=\{1,1,1,1\}$, which can be combined into $\epsilon_{\xi}=\{1,1,2,1\}$, both have disjointedness level of $\Delta_\xi=\Delta_{\xi+1}=1$ and the heterogeneity degree is explicit, i.e., $N_{r_\xi}=2$. In contrast, the sub-SFC allocated in R1 (DC4) and sub-SFC distributed over DC1, DC2 and DC3 can not be combined into one reliability class as their levels of disjointedness are different, i.e., $2$ and $4$, respectively.
Fig. \ref{parallel} shows $n=4$ sub-SFCs from different reliability classes each, i.e., $N=4$, whereby each sub-SFC is placed in different way and has different level of disjointedness $\Delta_\xi$. Moreover, the heterogeneity degree $N_r$ can not be explicitly defined providing overall value of $2$, i.e., VNFs of the same type are separated at least by VMs and servers. As can be seen, some sub-SFCs from different reliability classes utilize the same component, i.e., R2 and DC7, which we define as follows: The components of any hierarchy level jointly utilized by sub-SFCs from different reliability classes are referred to as the \textbf{common root components}. 
For instance, in Fig. \ref{parallel}b), the sub-SFCs of reliability class $\xi=1,2$ have two common roots, i.e., R1 and DC7. Obviously, that any common root is a special shared component, i.e., a failure of any common root results in failure of multiple different reliability classes, i.e., sub-SFCs. Moreover, DC7 is the common root of three reliability classes, i.e., $\xi=1,2,3$, while the sub-SFC of class $\xi=4$ is placed separately from other reliability classes and does not have any common roots with them. 

To describe the common root components of any type $c$ and all related reliability classes combined by this component $c$, we introduce a set $\Phi=\{c_{w_1}, c_{w_2}, ..., c_{w_\rho}, ...\}$, where any $c_{w_\rho}$ describes a component type and each $w_\rho$ presents a set of indexes of reliability classes combined by the component $c_{w_\rho}$, i.e., $w_\rho=\{\xi_1, \xi_2, ...\}$ and $c_{\xi_1}=c_{\xi_2}=...=c_{w_\rho}$. Since all components of type $\mathcal C$ are always disjoint, i.e., use different VMs, they can not be a common root component $c_{w_\rho}$ resulting in $1\leq c_{w_\rho}\leq\mathcal C-1$, i.e., DC, rack or server. Note, that using the component type $\mathcal C-2$, i.e., server, as a common root will result in the same placement strategies for all sub-SFCs, which utilize this root. Thus, in case all components of type $\mathcal C$ have the same component dependent reliability, i.e., $p_{\mathcal C_1}=...=p_{\mathcal C_\xi}=p_{\mathcal C_{\xi+1}}=...$, all sub-SFCs connected by this common root will be from the same reliability class and, thus, the common root can be considered as a shared component only. Considering the example from Fig. \ref{parallel}, the configuration set for the common roots can be determined as $\Phi=\{2_{w_1}, 1_{w_2}\}$ with $w_1=\{1, 2\}$ and $w_2=\{1, 2, 3\}$, whereby the presented parallelized and backup protected SFC can be considered as a complex service function chaining. Thus, a utilization of $n$ sub-SFCs with $\Psi$ VNFs each results in \textbf{complex service function chaining}, when these sub-SFCs belong to $N\leq n$ different reliability classes and, thus, include a different amount of heterogeneous and interdependent DCN components.

\begin{table}[]
\centering
\caption{Notation}\label{t1}
\vspace{-0.3cm}
\resizebox{\columnwidth}{!}{
\begin{tabular}{|l|l|}
\cline{1-2}
$k$  & number of parallel sub-flows and, thus, active sub-SFCs; \\
$r$ &  number of backup sub-SFCs;  \\
$n$ &  total number of pre-reserved sub-SFCs, i.e., $n=k+r$;  \\
$\Psi$ & a  number of VNFs in a SFC and, thus, sub-SFC; \\
$\mathcal C$& number of hierarchy levels, i.e., different component types;\\
$c$& a certain component type, $1\leq c\leq\mathcal C$;\\
$N$ & total number of reliability classes, i.e., $1\leq N\leq n$;\\
$\xi$ & any certain reliability class, $1\leq \xi \leq N$;\\
$n_\xi$& a number of SFCs of reliability class $\xi$, $1\leq \xi\leq N$\\
$\bar n_{1_{\xi}}$& number of components from the highest hierarchy of reliability class $\xi$;\\
$n_{c_{\xi}}$ & a number of components of type $c$ inside components of type $c-1$;  \\
$n_{1}$ & a number of DC, if $\mathcal C=4$;  \\
$n_{2}$ & a number of racks inside a DC, if $\mathcal C=4$;  \\
$n_{3}$ & a number of servers inside a rack, if $\mathcal C=4$;  \\
$n_{4}$ & a number of VNFs of certain type inside a server, if $\mathcal C=4$;  \\
$p_{c_{\xi}}$ & component reliability of any component of type $c$ utilized to allocate class $\xi$;  \\
$p_{4}$ & the reliability of VNF for $\mathcal C=4$;  \\
$p_{3}$ &  the reliability of server for $\mathcal C=4$;  \\
$p_{2}$ & the reliability of rack for $\mathcal C=4$; \\ 
$p_{1}$ & the reliability of DC for $\mathcal C=4$;   \\
$N_r$ & number of component types to disjointedly place VNFs of the same type;\\
$\Delta$ & number of component types to disjointedly place VNFs of different type;\\
\cline{1-2}
\end{tabular}
} \vspace{-0.6cm}
\end{table}

\section{SFC Reliability Analysis}
Without loss of generality, we define service reliability as a probability that at least $k\geq1$ sub-flows can successfully traverse $k$ out of $n=k+r$ sub-SFCs to complete the requested service, where $r\geq0$. Since each sub-SFC consists of $\Psi$ different VNF types, at least $k\Psi$ VNFs of each type have to be available during service run-time. Our reliability analysis assumes that any component type $c$ utilized to allocate $\Psi$ VNFs of a certain reliability class $\xi$ have an equal component dependent reliability as previously discussed.

When backup sub-SFCs are placed so that there are $r$ backup VNFs, any components from any hierarchy level $c$ and any reliability class $\xi$ can fail as long as they do not result in a number of failed VNFs which is larger than $r$. However, there is a need to take into account VNF placement strategy and the interrelation between involved components of different and same hierarchy level. To be able to control the amount of components of any type and from any reliability class, which can fail without affecting the overall service maintenance, we utilize a parameter \textbf{Acceptable Component Failures (ACF)} $A_{c_\xi}$, i.e., the amount of failed hardware and software components that do not lead to interruption of the service. It is obvious that when no backup or other protection is applied, there is no ACF and all active components, i.e., DCs, racks, servers and VNFs, must be available during service run-time. Without loss of generality, the ACF strongly depends on the number of backup components and the placement of VNFs from a certain reliability class inside DCs, racks and servers. Additionally, we introduce a variable $\Lambda_{c_\xi}$, which describes the \textbf{number of remaining available components} of a certain type $c$, which were not affected by occurred failures of any other components, but can still fail during service run-time. Using these parameters, we next derive the placement independent and placement dependent end-to-end service, i.e., SFC, reliability, which is generally the probability that at least $\Lambda_c-A_c$ out of $\Lambda_c$ components of type $c$ are available and, thus, described by a well known Binomial formula determined as $\sum_{f_{c}=0}^{A_{c}}p(\Lambda_c, f_c, p_c)=\sum_{f_{c}=0}^{A_{c}}\binom{\Lambda_{c}}{f_{c}}p_{c}^{\Lambda_{c}-f_{c}}(1-p_{c})^{f_{c}}$, where $p(\Lambda_c, f_c, p_c)$ is a probability mass function of binomial distribution. The utilized notations are introduced in Table \ref{t1}.

\subsection{Placement independent SFC reliability}\label{indep}
In DCNs, where only components from the lowest hierarchy level $c=\mathcal C=4$, i.e., VNFs, can fail and all other components from hierarchy level from $1$ to $\mathcal C-1$, i.e., DCs, racks and servers, have reliability of $100\%$, i.e., $p_1=p_2=p_3=1$ for any reliability class, the end-to-end service reliability is independent of VNF placement strategy and the service is successful if at least $k$ out of $n$ VNFs from any reliability class are available. 
\begin{lemma}\label{Lem1}
The placement independent SFC reliability is only a function of the amount of reliability classes $N$, SFC length $\Psi$, ACF for any reliability class $\xi$ and the amount of available components $\Lambda$, i.e.,
\begin{equation}\label{Rindep}
R(N)=\left[\prod_{\xi=1}^{N}\sum_{f_{4_\xi}=0}^{A_{4_\xi}}\!\!\!p(\Lambda_{4_\xi}, f_{4_\xi}, p_{4_\xi})\right]^{\Psi}
\end{equation}
, where $A_{4_\xi}=min\{n_\xi, r-\sum_{\xi=1}^{\xi-1}f_{4_\xi}\}$ and $\Lambda_{4_\xi}=n_\xi$.
\end{lemma}
\begin{proof}
Let us first assume that a parallelized SFC consists of $\Psi=1$ VNF types and all $n$ VNFs have the same VNF reliability $p_4<1$, i.e., there is only one reliability class ($N=1$). Thus, the SFC reliability follows the parallel reliability model described by Binomial formula for calculation of the probability that at least $k$ out of $\Lambda_{4}=n$ components are available, i.e., $R(N=1)=\sum_{f_{4}=0}^{A_{4}}p(\Lambda_4, f_4, p_4)$.

When, for instance, the reliability of all active VNFs and all backup VNFs of a certain type is $p_{4_1}$ and $p_{4_2}$, respectively, and $p_{4_1}\neq p_{4_2}$, there are two reliability classes, $N=2$. Then, the SFC reliability is determined as a probability that at least $k$ over all $\Lambda_{4_1}+\Lambda_{4_2}=n_1+n_2=k+r=n$ VNFs do not fail and at least $k$ sub-SFCs can be composed to serve $k$ sub-flows. Since we assumed SFC of one VNF, $\Psi=1$, at least $k-f_{4_1}$ active VNFs, where $0\leq f_{4_1}\leq min\{k,r\}$, and at least $f_{4_1}$ backup VNFs have to be available, while also $f_{4_2}$ backup VNFs can fail without service interruption, where $0\leq f_{4_2}\leq r-f_{4_1}$. Thus, AVF of active and backup VNFs is determined as $A_{4_1}=min\{k,r\}$ and $A_{4_2}=min\{r, r-f_{4_1}\}=r-f_{4_1}$, respectively. Following the Binomial formula and, additionally, applying the serial reliability model as VNFs from both reliability classes, i.e., at least $k$, have to be available at the same time; the SFC reliability is determined as $R_4(N=2)=\sum_{f_{4_1}=0}^{A_{4_1}}p(\Lambda_{4_1}, f_{4_1}, p_{4_1})\sum_{f_{4_2}=0}^{A_{4_2}}p(\Lambda_{4_2}, f_{4_2}, p_{4_2})$, when $\Psi=1$, $\Lambda_{4_1}=k$, $\Lambda_{4_2}=r$ and $N=2$. 

It is obvious, that this example can be extended to $N\leq k+r$ reliability classes, whereby it is required to extend the equation above to up to $N$ summations. In this case, the AVF of any reliability class $\xi$ is calculated as $A_{4_\xi}=min\{n_\xi, r-\sum_{\xi=1}^{\xi-1}f_{4_\xi}\}$ to ensure that the number of failed VNFs over all reliability classes, i.e., from $1$ to $\xi-1$, is not larger then the number of backup VNFs provided. 
Thus, the SFC reliability for $\Psi=1$ and $N\geq1$ can be generalized as
\begin{equation}\label{RelNPsi1}
\begin{split}
R'(N)=&\prod_{\xi=1}^{N}\sum_{f_{4_\xi}=0}^{A_{4_\xi}}\!\!\!p(\Lambda_{4_\xi}, f_{4_\xi}, p_{4_\xi})=
\\&\!\!\!\!\!\!\sum_{f_{4_1}=0}^{A_{4_1}}\!\!\!p(\Lambda_{4_1}, f_{4_1}, p_{4_1})...\sum_{f_{4_N}=0}^{A_{4_N}}p(\Lambda_{4_N}, f_{4_N}, p_{4_N})
\end{split}
\end{equation}
, where $R'(N)$ is a probability that there are at least $k$ available out of $n$ VNFs of a certain type (VNF1) over all $N$ reliability classes.

Next, let us assume that the parallelized SFC consists of $\Psi>1$ different VNFs. Then, the SFC reliability is determined as a probability that at least $k$ over all $k+r=n$ VNFs of each type out of $\Psi$ types do not fail and at least $k$ sub-SFCs can be composed by $k\Psi$ VNFs of each type to serve $k$ sub-flows. When all VNFs of the same type, have the same VNF reliability, e.g., all VNF1 and all VNF2 have reliability $p_{4^1}$ and $p_{4^2}$, respectively, the SFC reliability is a probability that at least $k$ VNFs of each type, e.g., $k$ VNF1 and $k$ VNF2, are available as described by $R(N=1)$ above. Applying the serial reliability model, the SFC reliability is $R^{1...{\Psi}}=R^1(N=1)\cdot R^2(N=1)...R^{\Psi}(N=1)=\prod_{\psi=1}^{\Psi}  R^\psi(N=1)=[R(N=1)]^\Psi$. In case all VNFs of any type, i.e., from VNF1 to VNF$\Psi$, utilized to build a certain SFC of reliability class $\xi$ have the same VNF reliability value, i.e., $p_{4_\xi^1}=p_{4_\xi^2}=...=p_{4_\xi^{\Psi}}=p_{4_\xi}$, at most $A_4=r$ VNFs of any type can fail without any impact on service maintenance. Then, the service reliability can be generalized following the serial reliability model for the whole SFC of $\Psi$ VNFs as 
\begin{equation}\label{Rpower}
R=[R(N=1)]^{\Psi}=\left[\sum_{f_4=0}^{A_4}p(\Lambda_{4}, f_{4}, p_{4})\right]^{\Psi}
\end{equation}
Similarly, when there are $N$ reliability classes and any reliability class $\xi$ contains $n_\xi$ sub-SFCs, all $k_\xi\Psi$ VNFs from the same reliability class have the same VNF reliability $p_{4_\xi}$, the overall SFC reliability can be derived from Eq. \eqref{RelNPsi1} as $R(N, \Psi)=R'^{1...\Psi}(N)=\prod_{\psi=1}^{\Psi}R'^{\psi}(N)=[R'(N)]^{\Psi}$
, which results in Eq. \eqref{Rindep}.
\end{proof}

\subsection{Placement Dependent SFC reliability}
When any DCN component, i.e., DC, rack , server and VNF, can fail, i.e., $p_1<1$, $p_2<1$, $p_3<1$, the SFC reliability is a function of the amount of available and failed active and backup DCs, racks, servers and VNFs involved and their interdependence, which is a function of VNF placement strategy over Intra- and Inter-DCN. Let us further simplify notation as $P(\Lambda_c)=p(\Lambda_c,f_c,p_c)$ to reduce the size of some formulas provided below, where $1\leq c\leq \mathcal C$, $c\in\{4,3,2,1\}$. Since our main assumption is that all components from hierarchical level $c$ utilized to allocate the certain sub-SFC of reliability class $\xi$ have the same component reliability $p_{c_\xi}=p_c$, all $\Psi$ different VNFs from the certain sub-SFC of reliability class $\xi$ have the same VNF component reliability, i.e., $p_{4_{\xi}^{1}}=p_{4_{\xi}^{2}}=...=p_{4_{\xi}^{\Psi}}=p_{4_{\xi}}$, whereby all servers, racks and DC required to allocate this sub-SFC from class $\xi$ have the same reliability as well. In other words, all $\bar n_{1_{\xi}}$ DCs, all $(\bar n_{1_{\xi}} \cdot n_{2_{\xi}})$ racks and all $(\bar n_{1_{\xi}}\cdot n_{2_{\xi}}\cdot n_{3_{\xi}})$ servers have the same component reliability $p_{1_{\xi}}$, $p_{2_{\xi}}$ and $p_{3_{\xi}}$, respectively.

\subsubsection{One Reliability Class}
There is only one reliability class $N=1$, when all $n$ active and backup sub-SFCs are placed in the same fashion and all components of each type $c$ have the same component dependent reliability, i.e., $p_c$.  To allocate all $n$ VNFs of a certain type, e.g., VNF1, $1\leq \bar n_1\leq n$ active and backup DCs, $1\leq n_2\leq n$ active or backup racks inside any DC, i.e., in total $(\bar n_1)n_2$ active and backup racks, and $1\leq n_3\leq n$ active or backup servers inside any rack, i.e., in total $(\bar n_1)n_2n_3$ active and backup servers, and $(\bar n_1)n_2n_3n_4=(k+r)$ VNFs are required,  whereby any server can allocate $1\leq n_4\leq n$ VNFs of the same type. Similarly, the same amount of DC, racks and servers is utilized to allocate each other VNF type from the same sub-SFC, i.e., VNF2, ..., VNF$\Psi$. Thus, the number of components from hierarchy level $c$ allocated to the SFC can be calculated as $\bar n_c$, if $c=1$, and as $\bar n_1\prod_{c'=2}^{c}n_{c'}$, while some of this components, i.e., $\Delta$ different component types from the lower hierarchy level, are utilized for a certain VNF type only, other $\mathcal C-\Delta$ component types are the shared components and used to allocate VNFs of different types. Thus, there is a need to consider the VNF placement strategy and resulting level of disjointedness $\Delta$, i.e., to take into account the components' interdependency.

\begin{lemma}
When there is one reliability class, $N=1$, the placement dependent SFC reliability is a function of the disjointedness level $\Delta$, SFC length $\Psi$, ACF $A_c$ and the amount of available components $\Lambda_c$ for any component type $c$, i.e.,
\begin{equation}\label{RDelta}
\begin{split}
R_{_{\Delta}}&=\prod_{c=1}^{\mathcal C-\Delta}\sum_{f_c=0}^{A_c}\!\!P(\Lambda_c)\left[\prod_{c=\mathcal C-\Delta+1}^{\mathcal C}\sum_{f_c=0}^{A_c}\!\!P(\Lambda_c)\right]^\Psi
\end{split}
\end{equation}
, where $\prod_{c=1}^{\mathcal C-\Delta}\sum_{f_c=0}^{A_c}\!\!P(\Lambda_c)=1$, if $\Delta=\mathcal C$.
\end{lemma}

\begin{proof}
Let us first assume that SFC consists of $\Psi=1$ VNF types, e.g., VNF1 only, while all $n=k+r$ active and backup VNFs have the same VNF components reliability $p_4$ as well as all DCs, racks and servers involved to allocate these $n$ VNFs have the same component reliability $p_1$, $p_2$ and $p_3$, respectively. Then, the SFC reliability can be calculated with Binomial formula and the serial reliability model as previously discussed in Sec. \ref{indep}, whereby it is required to consider the availability of all component types involved $1\leq c\leq \mathcal C$, i.e., a minimal amount of DC, racks and servers is required to ensure availability of at least $k$ out of $n$ VNFs, i.e.
\begin{equation}\label{R-Psi1}
\begin{split}
&R(\Psi=1)=\prod_{c=1}^{\mathcal C}\sum_{f_c=0}^{A_c}p(\Lambda_c,f_c,p_c)=\sum_{f_1=0}^{A_1}p(\Lambda_1,f_1,p_1)\cdot \\&\sum_{f_2=0}^{A_2}p(\Lambda_2,f_2,p_2)\cdot
\sum_{f_3=0}^{A_3}p(\Lambda_3,f_3,p_3)\cdot
\sum_{f_4=0}^{A_4}p(\Lambda_4,f_4,p_4)
\end{split}
\end{equation}
To extend Eq. \eqref{R-Psi1} to arbitrary SFC length, i.e., $\Psi\geq 1$ VNFs, there is a need to consider the VNF placement strategy and resulting level of disjointedness $\Delta$. Since our placement strategies are based on SFC or VNF type, $\mathcal C-\Delta$ component types are the shared components and allocate a whole sub-SFC. Thus, the service reliability calculation depends on the level of disjointedness, i.e., $\Delta$, and the serial reliability model can be only applied to $R(\Psi=1)$, similarly as for Eq. \eqref{Rpower}, when each VNF type of a certain sub-SFC is allocated by different components of any hierarchy level $c$, $1\leq c\leq \mathcal C$, resulting in $\Delta =4$. In this case, any VNF type is placed in different DC, different racks, different servers and VMs and the service reliability is determined with Eq. \eqref{R-Psi1} as
$R_{_{\Delta=4}}=\prod_{\psi=1}^{\Psi}R(\Psi=1)=[R(\Psi=1)]^{\Psi}=\left[\sum_{f_1=0}^{A_1}P(\Lambda_1)\sum_{f_2=0}^{A_2}\!\!P(\Lambda_2) \sum_{f_3=0}^{A_3}\!\!P(\Lambda_3)\sum_{f_4=0}^{A_4}\!\!P(\Lambda_4)\right]^{\Psi}$.
When different VNF types, i.e., sub-SFCs, are placed in the same DC but different racks, servers and VMs ($\Delta=3$), the failure of DC, i.e., the shared component, affects all $\Psi$ VNFs of a certain sub-SFC and, thus, needs to be considered once and not separately for each out of $\Psi$ VNFs. Thus, the SFC reliability can be calculated by separation the probability that there is at least $\Lambda_1-A_1$ available DCs required for maintenance of at least $k$ sub-SFCs, while the serial reliability model is still valid:
$R_{_{\Delta=3}}=\!\!\sum_{f_1=0}^{A_1}\!\!P(\Lambda_1)\left[\sum_{f_2=0}^{A_2}\!\!P(\Lambda_2)\!\! \sum_{f_3=0}^{A_3}\!\!P(\Lambda_3)\!\!\sum_{f_4=0}^{A_4}\!\!P(\Lambda_4)\right]^{\Psi}$.
Similarly, when different VNF types, i.e., sub-SFCs, are placed in the same DC, same rack and different servers resulting in disjointedness level $\Delta=2$, there are $2$ shared components, DC and rack. Thus, any failure of DC or rack will affect all $\Psi$ VNFs and needs to be considered once to calculate the probability that at least $\Lambda_1-A_1$ DCs and $\Lambda_2-A_2$ racks are available to build $k$ sub-SFC of $\Psi$ VNFs each. Then, the SFC reliability is determined as $R_{_{\Delta=2}}=\!\!\sum_{f_1=0}^{A_1}\!\!P(\Lambda_1)\!\!\sum_{f_2=0}^{A_2}\!\!P(\Lambda_2)\left[\sum_{f_3=0}^{A_3}\!\!P(\Lambda_3)\!\!\sum_{f_4=0}^{A_4}\!\!P(\Lambda_4)\right]^{\Psi}$.
In case of VNF disjointedness only ($\Delta=1$), i.e., different VNF types are placed in the same server and, thus, in the same rack and same DC resulting in three shared component types, a failure of server, rack or DC will lead to a failure of a whole sub-SFC, i.e., to a failure of all $\Psi$ VNFs. Thus, only VNFs itself are independent from each other resulting in SFC reliability defined as $R_{_{\Delta=1}}=\sum_{f_1=0}^{A_1} P(\Lambda_1) \sum_{f_2=0}^{A_2} P(\Lambda_2)\cdot\\ \sum_{f_3=0}^{A_3}  P(\Lambda_3)\left[\sum_{f_4=0}^{A_4} P(\Lambda_4)\right]^{\Psi}.$ As a result, we can derive that the number of summations outside and inside the brackets, in the formulas for reliability calculation above, is a function of 
$\Delta$ and the number of shared components. In other words, the summations for the shared components are outside the brackets, i.e., $\mathcal C-\Delta$ summations, and the $\Delta$ summations for component types whose failure impacts one VNF type only, are placed inside the brackets. Thus, without loss of generality, the SFC reliability for any value $\Delta$ is determined by Eq. \eqref{RDelta}.
\end{proof}

\subsubsection{N-Reliability Classes}
Now, we consider the parallelized SFCs, whereby any out of $n$ active and backup sub-SFCs can belong to one out of $N$, $1\leq N\leq n$, reliability classes. Then, an index $\xi$, $1\leq\xi\leq N$, shows a certain reliability class of any active or backup component utilize to allocate the sub-SFCs from this reliability class $\xi$. Thus, any component type $c$ utilized to allocate reliability class $\xi$ and its component reliability are noted as $c_\xi$ and $p_{c_{\xi}}$, respectively. Multiple sub-SFCs can belong to a certain reliability class $\xi$, if these sub-SFCs are placed following a certain placement strategy and all $\mathcal C$ involved component types have a certain component reliability value $p_{c_{\xi}}$ each. Note that all components of hierarchy level $c$ utilized to allocate any $\Psi$ different VNFs of a certain SFC reliability class $\xi$ have the same component reliability, i.e., $p_{c^1_{\xi}}=...=p_{c^{\Psi}_{\xi}}=p_{c_{\xi}}$. 
Generally, there are $\Psi\geq1$ VNFs per sub-SFC, an arbitrary number of involved component types ($\mathcal C$ types) to place this VNFs and an arbitrary number of reliability classes ($N$ active and backup classes), whereby each reliability class has its own level of disjointedness, $1\leq\Delta_\xi\leq \mathcal C$. The expression for the service reliability needs to consider the level of disjointedness $\Delta_\xi$ of any reliability class $\xi$. The level of disjointedness $\Delta_\xi$ take into account a certain placement strategy and describes the amount of disjoint and shared components, i.e., shared by different VNF types. 
Additionally, some sub-SFCs from different reliability classes can have a common root components, e.g., are allocated in the same DC or in the same rack. The failure of these common root components result in failures of whole sub-SFCs and related component types belonging to different reliability classes. All  reliability classes combined by the common root and related root component are collected in set $\Phi=\{c_{w_1}, c_{w_2}, ..., c_{w_\rho}, ...\}$, where any $c_{w_\rho}$ describes a component type of common root, and each $w_\rho$ is a set of indexes of the reliability classes combined by the common root component of type $c_{w_\rho}$, i.e., $w_\rho=\{\xi_1, \xi_2, ...\}$ and $c_{\xi_1}=c_{\xi_2}=...=c_{w_\rho}$. 
Then, the resulting expression for service reliability needs to take into account $N\geq 1$ reliability classes of SFCs of length $\Psi\geq1$, the shared and common root components utilized for allocation of a sub-SFCs from a certain reliability class $\xi$, i.e., 

\begin{lemma}
When there is multiple reliability classes, $N\geq1$, the placement dependent SFC reliability is a function of the disjointedness level $\Delta$, SFC length $\Psi$, ACF $A_c$ and the amount of available components $\Lambda_c$ for any component type $c$, , whereby the shared and common root components utilized for allocation of sub-SFCs from a certain reliability class $\xi$ need to be taken under consideration, i.e.,
\begin{equation}\label{RNPsiRoot}
\begin{split}
&R_{\Delta_{1},...,\Delta_{N}}(\Phi)=\prod_{\rho=1}^{|\Phi|}\sum_{f_{c_{w_{\rho}}}=0}^{A_{c_{w_{\rho}}}}P(\Lambda_{c_{w_{\rho}}})\cdot\\&
\cdot\prod_{\xi=1}^{N}\!\prod_{c_\xi=1\atop \phi(c_\xi)=1}^{\mathcal C_\xi-\Delta_\xi}\!\!\!\sum_{f_{c_\xi}=0}^{A_{c_\xi}}\!\!\!\!P(\Lambda_{c_\xi})\bigg[\!\!\prod_{\xi=1}^{N}\prod_{c_\xi=\mathcal C_\xi-\Delta_\xi+1\atop \phi(c_\xi)=1}^{\mathcal C_\xi}\sum_{f_{c_\xi}=0}^{A_{c_\xi}}\!\!\!\!P(\Lambda_{c_\xi})\!\bigg]^\Psi\!\!
\end{split}
\end{equation}
, where the function $\phi(c_\xi)$ ensures that the common root components for any reliability class $\xi\in w_\rho$ are considered only once by the first summation over $f_{c_{w_{\rho}}}\equiv f_{c_\xi}$, i.e., 
\begin{equation}\label{RootCntr}
\phi(c_\xi)=\begin{cases}
1, \!\!\!\!\!\!& \text{ if } \forall \rho: (\xi\in w_\rho \cap c_\xi\neq c_{w_\rho})\cup(\xi\notin w_\rho);\\
0,\!\!\!\!\!\! & \text{ else.}
\end{cases}
\end{equation}
\end{lemma}

\begin{proof}
First, we assume that the components from different reliability classes do not have any common root components, which corresponds to Inter-DCN placement of active and backup sub-SFCs. Thus, the SFCs of a certain reliability class $\xi$ are placed disjointedly from any other reliability classes. Assuming that any SFC contains only one VNF, $\Psi=1$, additionally relax the placement dependency as there can not be any shared components. Then, the service reliability can be determined based on Eqs. \eqref{RelNPsi1} and \eqref{R-Psi1} as follows
\begin{equation}\label{RN}
\begin{split}
R(\Psi=1, N)=&\prod_{\xi=1}^{N}R(\Psi=1)=\prod_{\xi=1}^{N}\prod_{c_\xi=1}^{\mathcal C}\sum_{f_{c_\xi}=0}^{A_{c_\xi}}\!\!p(\Lambda_{c_\xi},f_{c_\xi},p_{c_\xi})\!\!
\end{split}
\end{equation}
, where failure and availability of each component type and each reliability class are described by own summation. Thus, all $\mathcal C$ component types, i.e., from type $1$ to type $\mathcal C$ are considered for each reliability class $\xi$, $1\leq \xi\leq N$. Thus, any summation determines the probability that service will not fail in case of $f_{c_{\xi}}$ failures of any component type $c$ of a reliability class $\xi$, whereby at most $A_{c_{\xi}}$ failures of this component type $c$ and class $\xi$ are allowed. 

Let us next assume that some sub-SFCs from different reliability classes have a common roots, while there is not any shared component types, i.e., SFC length $\Psi=1$. The failure of the common root components result in failures of whole sub-SFCs and related component types belonging to different reliability classes. All  reliability classes combined by the common root and related root component are collected in set $\Phi$, while any failures of any component type $c_{w_\rho}$ related to the reliability classes from set $w_\rho$ affect ACF $A_{c_\xi}$ and the amount of available components $\Lambda_{c_\xi}$ of all related classes $\xi$, i.e., $\xi\in w_\rho$. Thus, to take into account the common roots for different reliability classes, Eq. \eqref{RN} for $\Psi=1$ is modified as
\begin{equation}\label{RNroots}
\begin{split}
&R(\Psi=1, N, \Phi)
=\prod_{\rho=1}^{|\Phi|}\sum_{f_{c_{w_{\rho}}}=0}^{A_{c_{w_{\rho}}}}\!\!P(\Lambda_{c_{w_{\rho}}})\!\!\prod_{\xi=1}^{N}\prod_{c_\xi=1\atop \phi(c_\xi)=1}^{\mathcal C}\sum_{f_{c_\xi}=0}^{A_{c_\xi}}\!\!P(\Lambda_{c_\xi})\!\!
\end{split}
\end{equation}
, where $\Lambda_{c_{w_{\rho}}}=1=const$ as there can be only one common root $c_{w_\rho}$ for reliability classes from set $w_\rho$. The function $\phi(c_\xi)$ determined by Eq. \eqref{RootCntr} ensures that the common root components for any reliability class $\xi\in w_\rho$ are considered only once, i.e., by the first summation over $f_{c_{w_{\rho}}}$.

Next, let us assume that SFC consists of $\Psi\geq1$ VNFs and there is no common root components, but some sub-SFCs utilize shared components to place $\Psi$ different VNFs. Thus, each reliability class $\xi$ has its own level of disjointedness $\Delta_\xi$, $1\leq\Delta_\xi\leq \mathcal C$, and the expression for the service reliability needs to consider the level of disjointedness $\Delta_\xi$ of any reliability class $\xi$, $1\leq\xi\leq N$. Only considering the reliability class $\xi$ which provides the level of disjointedness $\Delta_\xi$ due to a certain placement strategy, there are $\mathcal C-\Delta_\xi$ shared components, i.e., shared by different VNF types, i.e., components $c_\xi=1_\xi, ...,(\mathcal C-\Delta_\xi)_\xi$, and $\Delta_\xi$ disjoint components, i.e., $c_\xi=(\mathcal C-\Delta_\xi+1)_\xi,...,\mathcal C_\xi$. Then, similar to Eq. \eqref{RDelta}, there are $\mathcal C_\xi$ summations for each reliability class $\xi$, where all summations for shared components and all summations for disjoint components should be separated as each summation for disjoint components is valid for each out of $\Psi$ VNF types. Thus, the service reliability is
\begin{equation}\label{RNPsi}
\begin{split}
&R_{\Delta_{1},...,\Delta_{N}}\!\!=\!\!\prod_{\xi=1}^{N}\!\prod_{c_\xi=1}^{\mathcal C_\xi-\Delta_\xi}\!\!\!\sum_{f_{c_\xi}=0}^{A_{c_\xi}}\!\!\!\!P(\Lambda_{c_\xi})\bigg[\prod_{\xi=1}^{N}\prod_{c_\xi=\mathcal C_\xi-\Delta_\xi+1}^{\mathcal C_\xi}\sum_{f_{c_\xi}=0}^{A_{c_\xi}}\!\!\!\!P(\Lambda_{c_\xi})\bigg]^\Psi\!\!\!\!\!\!
\end{split}
\end{equation}
, where any $A_{c_{\xi}}$ and $A'_{c_{\xi}}$ describes ACF of a shared and disjoint component type $c_\xi$ of any reliability class $\xi$, respectively.
When level of disjointedness is maximal, i.e., $\Delta_{\xi}=\mathcal C$, all component types are disjoint and the same for all $\Psi$ VNFs and related components utilized for their allocation, i.e., all $\mathcal C$ related summations are inside the brackets in Eq. \eqref{RNPsi}, i.e., $\prod_{c=1}^{\mathcal C-\Delta}\sum_{f_c=0}^{A_c}\!\!P(\Lambda_c)=1$, if $\Delta=\mathcal C$. Similar, in case of $\Delta_{\xi}=1$, only components of one type, i.e., $c=\mathcal C$, are disjoint resulting in only one summation per reliability class $\xi$ inside brackets. When $\Delta_{\xi}=\mathcal C-1$, i.e., only component type $1$, i.e., DCs, is a shared component for all $\Psi$ different VNF types of reliability class $\xi$, there is one summation outsides the brackets, which describes the probability that there are enough available components of type $c=1$ to maintain the service.

Finally, when there are some shared components for VNFs from sub-SFCs of length $\Psi\geq1$ as well as common roots for different reliability classes described by set $\Phi$, we can derive the service reliability with Eqs. \eqref{RNPsi} and \eqref{RNroots}. The resulting reliability equation takes into account $N\geq 1$ reliability classes of SFCs of length $\Psi\geq1$, their common root and shared components and failures of any out of $\mathcal C$ component types involved for allocation of a sub-SFCs from a certain reliability class $\xi$, i.e., 
$R_{\Delta_{1},...,\Delta_{N}}(\Phi)=\prod_{\rho=1}^{|\Phi|}\sum_{f_{c_{w_{\rho}}}=0}^{A_{c_{w_{\rho}}}}\!\!P(\Lambda_{c_{w_{\rho}}})R_{\Delta_{1},...,\Delta_{N}}$
and presents Eq. \eqref{RNPsiRoot}.

Finally, let us derive Eq. \eqref{RootCntr}, where $\phi(c_\xi)$ needs to ensure that the common root components are consider only once by the first summation in Eqs. \eqref{RNPsiRoot} and \eqref{RNroots} and only component types $c_\xi$ which are not the common root components are considered in the following summation to calculate the probability that there is enough available components of type $c_\xi$ to maintain the service. As a result, $c_\xi$ needs to be additionally considered, when either the reliability class $\xi$ of considered component type $c_\xi$ does not have a common root components with other reliability classes, i.e., $\forall \rho: \xi\notin w_\rho$, or $\xi$ have some common root components with other reliability classes, i.e., $\xi\in w_\rho$, but the considered component type $c$ is not the common root, i.e., $c_\xi\neq c_{w_\rho}$.
\end{proof}

\subsubsection{Acceptable Component Failures}
The bound of summations $A_{c_{w_{\rho}}}$ and $A_{c_{\xi}}$ in Eq.  \eqref{RNPsiRoot} is what we refer to as the acceptable component failure, which is a function of pre-reserved components of type $c$ from same reliability class $\xi$ and the failures of any other components of higher hierarchy level or failure of preserved components of the lowest hierarchy level over all reliability classes. 
Depending on the placement strategy and, thus, the reliability class $\xi$, there is $\bar n_{1_\xi}$ active and backup components of the highest hierarchy level such as DCs and $n_{c_\xi}$ components of any type $c>1$ allocated inside a component of type $c-1$. Thus, there are in total $\bar n_{1_\xi}\prod_{c'=2}^{c} n_{c'_\xi}$ components of type $c$ utilized for allocation of $n_\xi$ active and backup VNFs of certain type of any reliability class $\xi$.

Without loss of generality, each summation in Eq. \eqref{RNPsiRoot} is a function of the prior summation, i.e., any $A_{c_{\xi}}$ and $\Lambda_{c_{\xi}}$ are the functions of already considered component failures $A_{c_{\xi}}\equiv A_{c_{\xi}}(f_{1_1}, ..., f_{c_1}, ..., f_{\mathcal C_1},\\ ..., f_{1_\xi}, ..., f_{c_\xi}, ..., f_{\mathcal C_\xi}, ...)$ and $\Lambda_{c_{\xi}}\equiv \Lambda_{c_{\xi}}(f_{1_1}, ..., f_{c_1}, ..., f_{\mathcal C_1}, ...,\\ f_{1_\xi}, ..., f_{c_\xi}, ...)$, where the previously assumed failures reduce the acceptable failures of component $c_\xi$ and the amount of remaining available components of type $c$ from class $\xi$, respectively, where $1\leq\xi\leq N$. 
Next, we derive the amount of available components $\Lambda_{c_\xi}$ and ACF $A_{c_\xi}$, where $1\leq\xi\leq N$ and $1\leq c\leq \mathcal C$. 

Let us first derive ACF $A_{c_{w_{\rho}}}$ of the common root components of any reliability class $\xi$ from set $w_{\rho}$, i.e., $\forall \xi: \xi\in w_{\rho}$. Since there is only one common root component per $|w_{\rho}|$ reliability classes, the amount of available roots $\Lambda_{w_{\rho}}\leq1$ and the ACF of the root component is ether $0$ or $1$, both depend on the number of failed VNFs of each type in case of failure of the roots and failed roots from higher hierarchy level resulting in failures of other roots. Thus, the failure of a certain common root $c_{w_{\rho}}$ is only possible, when the other roots from the higher hierarchy do not fail leading to the failure of this root $c_{w_{\rho}}$, i.e., $\forall \rho\rq{}, 1\leq\rho\rq{}\leq \rho-1: f_{c_{w_{\rho}}}\leq1$, if $f_{c_{w_{\rho\rq{}}}}=0 \cap  w_{\rho}\subset w_{\rho\rq{}}$. In other words, the root component can only fail if it is available, i.e., not failed due to failures of other components, $\forall \rho\rq{}, 1\leq\rho\rq{}\leq \rho-1: \Lambda_{c_{w_{\rho}}}=1 \text{, if } f_{c_{w_{\rho\rq{}}}}=0 \cap  w_{\rho}\!\!\subset\!\! w_{\rho\rq{}}$. As a results, the amount of available common roots is defined as
\begin{equation}
\Lambda_{c_{w_{\rho}}}=\begin{cases}
1, \!\!\!\!\!\!& \text{ if } \sum\limits_{\rho\rq{}=1\atop w_{\rho}\subset w_{\rho\rq{}}}^{\rho-1}f_{c_{w_{\rho\rq{}}}}=0;\\
0,\!\!\!\!\!\! & \text{ else.}
\end{cases}
\end{equation}
Generally, to maintain the service, the amount of failed VNFs may not be larger than the number of backup VNFs, i.e., $r$. The amount of VNFs of a certain type belonging to reliability class $\xi\in w_{\rho}$ can be calculated as $n_{c_{w_{\rho}}}\prod_{c=c_{w_{\rho}}+1}^{\mathcal C} n_{c_\xi}=\prod_{c=c_{w_{\rho}}+1}^{\mathcal C} n_{c_\xi}$ VNFs of class $\xi$, which will fail with failure of the root component type $c_{w_{\rho}}$. Since there are $|w_{\rho}|$ reliability classes affected by failure of root component type $c_{w_{\rho}}$, the total number of failed VNFs of certain type from any reliability class $\xi\in w_{\rho}$ is determined as $\sum_{i=1}^{|w_{\rho}|}\prod_{c=c_{w_{\rho}}+1}^{\mathcal C} n_{c_{\xi_i}}$, where $\xi_i\in w_{\rho}$ and $\xi_i\equiv \xi$.  Since there can be other common roots from the higher hierarchy level which affects the considered root and other roots as well, there is a need to consider all other VNF failures due to failure of all other $\rho-1$ roots. Thus, the amount of failed VNFs of certain type over $\rho$ root components can be calculated as a function of any $f_{c_{w_{\rho}}}$, i.e., $\sum_{\rho\rq{}=1}^{\rho} f_{c_{w_{\rho\rq{}}}} \sum_{i=1 \atop \sigma=1}^{|w_{\rho\rq{}}|}\prod_{c=c_{w_{\rho\rq{}}}+1}^{\mathcal C} n_{c_{\xi_i}}$, where  $\sigma=1$, iff $\forall \rho\rq{}\rq{}, 1\leq\rho\rq{}\rq{}\leq\rho\rq{}-1: \xi_i\notin w_{\rho\rq{}\rq{}}$, and ensures that the failures of VNFs from the same reliability class are considered only once. The formula above take into account placement strategies, where, for instance, DC and rack inside this DC represent two common roots for the same reliability classes, whereby DC can combine more reliability classes and lead to failure of rack if fails.  Summarizing all constraints above, the ACF of the common root components can be calculated as follows
\begin{equation}
A_{c_{w_{\rho}}}=\begin{cases}
1, \!\!\!\!\!\!& \text{ if } (r\geq\!\!\! \sum\limits_{\rho\rq{}=1}^{\rho} f_{c_{w_{\rho\rq{}}}} \!\!\!\sum\limits_{i=1 \atop \sigma=1}^{|w_{\rho\rq{}}|}\prod\limits_{c=c_{w_{\rho\rq{}}}+1}^{\mathcal C} n_{c_{\xi_i}}) \cap \Lambda_{c_{w_{\rho}}}=1;\\
0,\!\!\!\!\!\! & \text{ else.}
\end{cases}
\end{equation}
where  $\sigma=1$, iff $\forall \rho\rq{}\rq{}, 1\leq\rho\rq{}\rq{}\leq\rho\rq{}-1: \xi_i\notin w_{\rho\rq{}\rq{}}$.

The amount of available active and backup components $\Lambda_{c_{\xi}}$ of any component type $c_\xi$ for any reliability class $\xi$, after failure of any component of higher hierarchy level, i.e., of type $1$ to $c-1$, is
\begin{equation}\label{LACI}
\!\!\!\!\Lambda_{c_{\xi}}\!\!=\!\!\begin{cases}
\Lambda_{c_{w_{\rho}}}, \!\!\!\!\!\!& \text{ if } \forall \rho: (\xi\in w_\rho) \cap (c_\xi=c_{w_\rho});\\
\bar n_{\mathcal C_{\xi}},\!\!\!\!\!\! & \text{ if } \forall \rho: (\xi\notin w_\rho) \cap (c=1);\\
(\Lambda_{(c-1)_{\xi}}-f_{(c-1)_{\xi}})n_{c_{\xi}}, \!\!\!\!\!\!& \text{ if }\forall \rho: (\xi\notin w_\rho) \cap (1< c\leq\mathcal C).
\end{cases}
\end{equation}
, where the first case determined as $\bar n_{1_{\xi}}=n_{c_\xi}=n_{c_{w_\rho}}=\Lambda_{c_{w_{\rho}}}$, describes a number of the common root components of a certain type $c$ and is the same for any reliability class $\xi$ from set $w_\rho$. However, when there is one reliability class only, $N=1$, there is no common root components and the first case in Eq. \eqref{LACI} will be never true.
The derivation of Eq. \eqref{LACI} is presented in Appendix \ref{firstLambda}.

Generally, the number of any failed components $f_{c_{\xi}}$ of type $c_\xi$ can vary as $0\leq f_{c_{\xi}}\leq A_{c_{\xi}}$ for shared components and $0\leq f_{c_{\xi}}\leq A'_{c_{\xi}}$ for disjoint components from any certain reliability class $\xi$, whereby any $f_{c_{\xi}}$ out of $\Lambda_{c_{\xi}}$ components of type $c_\xi$ of reliability class $\xi$ can fail without interrupting the end-to-end service. Without loss of generality, ACF, i.e., $A_{c_{\xi}}$ for components shared by $\Psi$ different VNF types of class $\xi$ and $A'_{c_{\xi}}$ for disjoint components of the same class, is a function of available components $\Lambda_{c_{\xi}}$ of type $c$ and a reliability class $\xi$, the amount of provided backup VNFs $r$, and the number of VNFs considered as already failed $F_{c_\xi}$ after $f_{c_{\xi}}$ components of any type $c$, $1\leq c\leq\mathcal C$, and any class $\xi$, $1\leq\xi\leq\mathcal C$ failed. The amount of failed components $F_{c_{\xi}}$ of the lowest hierarchy level $\mathcal C$ and a certain type, e.g., VNFs1, due to failures ($f_{c'_\xi}$) of any component types $c'$, i.e., $c'\leq c$ belonging to a certain reliability class $\xi$ can be generalized to an iterative function as 
\begin{equation}\label{Fcomp}
F_{c_{\xi}}=\begin{cases}
F_{(c-1)_{\xi}}+f_{c_{\xi}}\prod\limits_{c'=c+1}^{\mathcal C}n_{c'_{\xi}}&\text{, if } 1\leq c\leq\mathcal C;\\
0&\text{, else. }
\end{cases}
\end{equation}
, where the failures of any component from the higher hierarchy level, from $1$ to $c$, are taken into account. These $F_{c_{\xi}}$ failed components of type $\mathcal C$ due to failures of $f_{c_{\xi}}$ components of the type from $1$ to $c$ reduce the amount of available backup components $r$, which are allowed to fail, i.e., reduces ACF. The failure of any common root components $c_{w_\rho}$ are taken into account by Eq. \eqref{Fcomp} as $\forall \rho: f_{c_{w_\rho}}\equiv f_{c_\xi}$, if $c_{w_\rho}=c_\xi$ and $\xi\in w_\rho$.

Then, ACF $A_{c_{\xi}}$ of any shared component, i.e., placed outside the brackets in Eq. \eqref{RNPsiRoot}, is a minimum between the total number of available VNFs provided by a certain reliability class $\xi$ and allocated inside component $c$, i.e., $\Lambda_{c_\xi}\prod_{c'=c_\xi+1}^{\mathcal C}n_{c'_{\!\xi}}$ and the total number of available backup VNFs $r$ reduced by any failures of any component type and from any reliability class. The result will show how many VNFs of a certain reliability class $\xi$ may fail without any affect of service maintenance.
The amount of backup VNFs is generally reduced by the VNF failures of class $\xi$ due to failures of the components from the higher hierarchy level, i.e., from $1$ to $c-1$, and from the same reliability class $\xi$, which can be calculated with Eq. \eqref{Fcomp} as $F_{(c-1)_{\xi}}$. Additionally, the number of backup VNFs is reduced through VNF failures due to failure of any shared component from other reliability classes with an index which is less than $\xi$. That failures can be also determined with Eq. \eqref{Fcomp} as $\sum_{l=1}^{\xi-1}F_{(\mathcal C-\Delta_l)_l}$, where $\mathcal C-\Delta_l$ defines the lowest hierarchy of a shared component related to the reliability class $l$. Thus, defining ACF for component $c_\xi$, we need to take into account that there are $r-\sum_{l=1}^{\xi-1}F_{(\mathcal C-\Delta_l)_l}-F_{(c-1)_{\xi}}$ remaining backup VNFs, which can fail without service interruption. Thus, the amount of VNFs from class $\xi$, which may fail are defined as $min\{\Lambda_{c_\xi}\prod_{c'=c_\xi+1}^{\mathcal C}n_{c'_{\!\xi}}; r-\sum_{l=1}^{\xi-1}F_{(\mathcal C-\Delta_l)_l}-F_{(c-1)_{\xi}}\}$. Since we are interested in ACF of component $c_\xi$, which can be from any hierarchy level, i.e., $1\leq c\leq \mathcal C$, there is a need for the normalization by the amount of VNFs allocated by component $c_\xi$, i.e., by $\prod_{c'=c+1}^{\mathcal C}n_{c'_{\!\xi}}$. This is also the number of VNFs, which will fail if one component $c_\xi$ fails.
As a result, ACF for a certain component type $c_\xi$ shared by $\Psi$ different VNFs of a certain reliability class $\xi$ can be derived as 
\begin{equation}\label{Ashar}
A_{c_{\xi}}=\left\lfloor\frac{min\left\{\Lambda_{c_\xi}\prod\limits_{c'=c+1}^{\mathcal C}n_{c'_{\!\xi}};r-\sum\limits_{l=1}^{\xi-1}F_{(\mathcal C-\Delta_l)_l}-F_{(c-1)_{\xi}}\right\}}{\prod\limits_{c'=c+1}^{\mathcal C}n_{c'_{\!\xi}}}\right\rfloor
\end{equation}
, where the sum of all $F_{(\mathcal C-\Delta_{l})_l}$ identifies all summations in Eq. \eqref{RNPsiRoot} located prior to summation with a bound $A_{c_{\xi}}$ and, finally, all prior failures of shared components.
Obviously, that if any prior reliability class $l<\xi$ have disjointedness level of $\Delta_l=\mathcal C$ all $\Psi$ VNFs as well as all related component types of this class $l$ are disjoint and there is no shared components and, thus, summations  outside the brackets for this reliability class in Eq. \eqref{RNPsiRoot}, i.e., $F(\Delta_{l})=0$. When $\Delta_l<\mathcal C$, a failure of $\mathcal C-\Delta_l$ components types of a certain class $l$ needs to be considered for calculation of ACF $A_{c_{\xi}}$, where $l<\xi$.
Thus, the function $F_{(\Delta_{\xi}-1)_{\xi}}$ provides a total number of the failed components of lowest hierarchy level due to failure of any component types up to $\Delta_\xi-1$.

Similar, ACF for any disjoint, i.e., unshared, component type $c_\xi$ and any reliability class $\xi$, i.e., $A'_{c_{\xi}}$ for summations inside the brackets in Eq. \eqref{RNPsiRoot}, is a function of the remaining available components $\Lambda_{c_{\xi}}$ of type $c$ from a reliability class $\xi$, the amount of provided backup VNFs $r$ and the number of VNFs considered as failed. Thus, there is a need to define the amount of VNFs which may still fail without interrupting the service, i.e., the minimum between the total number of available VNFs provided by the reliability class $\xi$ and allocated inside component $c_\xi$, i.e., $\Lambda_{c_\xi}\prod_{c'=1}^{c-1}n_{c'_{\!\xi}}$ and the amount of remaining backup VNFs after failures of any other components from any other reliability classes. Defining ACF of any disjoint component type $c_\xi$, the number of backup VNFs $r$ is already reduced by possible failures of components from hierarchy level $1$ to $\mathcal C$ of all reliability classes $l$, $1\leq l\leq \xi-1$. These failures can be calculated with Eq. \eqref{Fcomp} as $\sum_{l=1}^{\xi-1}F_{\mathcal C_l}$ failed VNFs from $\xi-1$ different reliability classes. Some VNFs from reliability class $\xi$ could fail due to failures of components from higher hierarchy level belonging to the same class $\xi$, the number of failed VNFs from reliability class $\xi$ is calculated with Eq. \eqref{Fcomp} as $F_{(c-1)_{\xi}}$, which additionally reduce the number of available backup VNFs. Moreover, the shared components of any reliability class $l$, $\xi<l\leq N$, i.e., these reliability classes are not yet considered inside the brackets in Eq. \eqref{RNPsiRoot}, can fail resulting in $\sum_{l=\xi+1}^{N}F_{(\mathcal C-\Delta_l)_l}$ additional VNF failures. Thus, the amount of backup VNFs of a certain type is determined as $r-\sum_{l=1}^{\xi-1}F_{\mathcal C_l}-F_{(c-1)_{\xi}}-\sum_{l=\xi+1}^{N}F_{(\mathcal C-\Delta_l)_l}$ and defines the maximal number of VNFs from the reliability class $\xi$ which may fail due to failures of component $c_\xi$, which can be reduced through value $\Lambda_{c_\xi}\prod_{c'=1}^{c-1}n_{c'_{\!\xi}}$. To calculate the amount of components $c_\xi$, which may fail resulting in allowed VNFs failures calculated above, the required normalization is performed similar to Eq. \eqref{Ashar} as $\prod_{c'=c+1}^{\mathcal C}n_{c'_{\!\xi}}$ to determine the number of VNFs allocated inside component $c_\xi$.
As a result, ACF for a certain component type $c_\xi$ utilized to disjointly allocate $\Psi$ different VNFs of a certain reliability class $\xi$ is
\begin{equation}\label{Adis}
A'_{c_{\xi}}\!\!=\!\!\left\lfloor\!\!\frac{min\left\{\!\!\Lambda_{c_\xi}\prod\limits_{c'=1}^{c-1}n_{c'_{\!\xi}};r\!-\!\!\sum\limits_{l=1}^{\xi-1}F_{\mathcal C_l}\!-\!F_{(c-1)_{\xi}}\!-\!\!\sum\limits_{l=\xi+1}^{N}F_{(\mathcal C-\Delta_l)_l}\right\}}{\prod\limits_{c'=c+1}^{\mathcal C}n_{c'_{\!\xi}}}\!\!\right\rfloor
\end{equation}

\begin{table*}[]
\centering
\caption{The placement dependent SFC reliability for $k=1$ vs. backup SFCs $r$ and SFC length $\Psi$}\vspace{-0.3cm}
\label{compare}
\resizebox{\textwidth}{!}{%
\begingroup
\setlength{\tabcolsep}{0.7pt} 
\renewcommand{\arraystretch}{1.2} 
\begin{tabular}{|ccc|c|c|c|c|c|c|c|c|c|c|c|c|}
\hline
\multirow{3}{*}{$N_r$} & \multirow{3}{*}{$\Delta$} & \multirow{3}{*}{Pl.} & \multicolumn{6}{c|}{Case 1} & \multicolumn{6}{c|}{Case 2} \\ \cline{4-15} 
 &  &  & \multicolumn{3}{c|}{$\Psi=4$} & \multicolumn{3}{c|}{$\Psi=8$} & \multicolumn{3}{c|}{$\Psi=4$} & \multicolumn{3}{c|}{$\Psi=8$} \\ \cline{4-15} 
 &  &  & $r=0$ & $r=1$ & $r=100$ & $r=0$ & $r=1$ & $r=100$ & $r=0$ & $r=1$ & $r=100$ & $r=0$ & $r=1$ & $r=100$ \\ \cline{4-15} 
4 & 1 & s & 0.9595298551 & 0.9995123444 & $\to 1$ & 0.9217205502 & 0.9990295230 & $\to 1$ & 0.9999300030 & 0.9999999963 & $\to 1$ & 0.9998900055 & 0.9999999935 & $\to 1$ \\
4 & 2 & s & 0.9566541431 & 0.9995075495 & $\to 1$ & 0.9152878303 & 0.9990157556 & $\to 1$ & 0.9999000045 & 0.9999999948 & $\to 1$ & 0.9998200153 & 0.9999999900 & $\to 1$ \\
4 & 3 & s & 0.9563671756 & 0.9995073731 & $\to 1$ & 0.9146473210 & 0.9990150262 & $\to 1$ & 0.9998700078 & 0.9999999939 & $\to 1$ & 0.9997500300 & 0.9999999879 & $\to 1$ \\
4 & 4 & v, s & 0.9563384848 & 0.9995073585 & $\to 1$ & 0.9145832976 & 0.9990149596 & $\to 1$ & 0.9998400120 & 0.9999999936 & $\to 1$ & 0.9996800496 & 0.9999999872 & $\to 1$ \\ \hline
3 & 1 & s & 0.9595298551 & 0.9995031488 & 0.9999899999 & 0.9217205502 & \textbf{0.9990210788} & 0.9999899999 & 0.9999300030 & 0.9999899976 & 0.9999899999 & 0.9998900055 & 0.9999899956 & 0.9999899999 \\
3 & 2 & s & 0.9566541431 & 0.9994984114 & 0.9999899999 & 0.9152878303 & 0.9990074399 & 0.9999899999 & 0.9999000045 & 0.9999899967 & 0.9999899999 & 0.9998200153 & 0.9999899935 & 0.9999899999 \\
3 & 3 & v, s & 0.9563671756 & 0.9994982407 & 0.9999899999 & 0.9146473210 & 0.9990067233 & 0.9999899999 & 0.9998700078 & 0.9999899964 & 0.9999899999 & 0.9997500300 & 0.9999899928 & 0.9999899999 \\
3 & 4 & v & 0.9563384848 & 0.9994682561 & 0.9999600006 & 0.9145832976 & 0.9989367949 & 0.9999200028 & 0.9998400120 & \textbf{0.9999599970} & \textbf{0.9999600006} & 0.9996800496 & \textbf{0.9999199956} & \textbf{0.9999200028} \\ \hline
2 & 1 & s & 0.9595298551 & 0.9994111840 & 0.9998900009 & 0.9217205502 & 0.9989366284 & 0.9998900009 & 0.9999300030 & 0.9999799988 & 0.9999800001 & 0.9998900055 & 0.9999799976 & 0.9999800001 \\
2 & 2 & v, s & 0.9566541431 & 0.9994070213 & 0.9998900009 & 0.9152878303 & 0.9989242748 & 0.9998900009 & 0.9999000045 & 0.9999799985 & 0.9999800001 & 0.9998200153 & 0.9999799969 & 0.9999800001 \\
2 & 3 & v & 0.9563671756 & 0.9991072291 & 0.9995900640 & 0.9146473210 & 0.9982252375 & 0.9991902879 & 0.9998700078 & \textbf{0.9999499994} & \textbf{0.9999500010} & 0.9997500300 & \textbf{0.9999100004} & \textbf{0.9999100036} \\
2 & 4 & v & 0.9563384848 & 0.9990772562 & 0.9995600766 & 0.9145832976 & 0.9981553639 & 0.9991203467 & 0.9998400120 & \textbf{0.9999200012} & \textbf{0.9999200028} & 0.9996800496 & \textbf{0.9998400088} & \textbf{0.9998400120} \\ \hline
1 & 1 & v, s & 0.9595298551 & 0.9984906149 & 0.9988901109 & 0.9217205502 & 0.9980912785 & 0.9988901109 & 0.9999300030 & 0.9999699999 & 0.9999700003 & 0.9998900055 & 0.9999699995 & 0.9999700003 \\
1 & 2 & v & 0.9566541431 & 0.9954981375 & 0.9958964363 & 0.9152878303 & 0.9911255646 & 0.9919188219 & 0.9999000045 & \textbf{0.9999400011} & \textbf{0.9999400015} & 0.9998200153 & \textbf{0.9999000037} & \textbf{0.9999000045} \\
1 & 3 & v & 0.9563671756 & 0.9951995179 & 0.9955976973 & 0.9146473210 & 0.9904319848 & 0.9912246871 & 0.9998700078 & \textbf{0.9999100032} & \textbf{0.9999100036} & 0.9997500300 & \textbf{0.9998300128} & \textbf{0.9998300136} \\
1 & 4 & v & 0.9563384848 & 0.9951696622 & 0.9955678297 & 0.9145832976 & 0.9903626567 & 0.9911553034 & 0.9998400120 & \textbf{0.9998800062} & \textbf{0.9998800066} & 0.9996800496 & \textbf{0.9997600268} & \textbf{0.9997600276} \\ \hline
\end{tabular}%
\endgroup
}\vspace{-0.2cm}
\end{table*}

\begin{table*}[h!]
\centering
\caption{Service reliability provided by Pl. with $N_r=4$ and $\Delta=1$ vs. a number of parallel active $k$ and backup $r$ sub-SFCs.}
\label{res}\vspace{-0.3cm}
\resizebox{0.78\textwidth}{!}{%
\begin{tabular}{|c|ccc|ccc|}
\hline
              & \multicolumn{3}{c|}{$\Psi=4$}                              & \multicolumn{3}{c|}{$\Psi=8$}                              \\\hline
$\frac{r}{k}$ & $k=1$             & $k=4$             & $k=8$             & $k=1$             & $k=4$             & $k=8$             \\\hline
0             & 0.959529855054096 & 0.847683965207723 & 0.718568104870289 & 0.921720550240843 & 0.721767099608634 & 0.520947746077460 \\
0.125         & -                 & -                 & 0.983696719311009 & -                 & -                 & 0.968371471691151 \\
0.25          & -                 & 0.995273874413365 & 0.999384198321856 & -                 & 0.990697185787667 & 0.998774974539633 \\
0.375         & -                 & -                 & 0.999981195954079 & -                 & -                 & 0.999962424513603 \\
0.5           & -                 & 0.999893621626573 & 0.999999500094456 & -                 & 0.999787784003278 & 0.999999000321022 \\
0.625         & -                 & -                 & 0.999999987999969 & -                 & -                 & 0.999999976000402 \\
0.75          & -                 & 0.999997932451063 & 0.999999999734000 & -                 & 0.999995866449611 & 0.999999999468002 \\
0.875         & -                 & -                 & 0.999999999994470 & -                 & -                 & 0.999999999988942 \\
1             & 0.999512344397461 & 0.999999963313851 & 0.999999999999890 & 0.999029523012867 & 0.999999926631523 & 0.999999999999781 \\
1.125         & -                 & -                 & 0.999999999999997 & -                 & -                 & 0.999999999999995 \\
1.25          & -                 & 0.999999999389706 & 0.999999999999999 & -                 & 0.999999998779421 & 0.999999999999998 \\
1.375         & -                 & -                 & 0.999999999999999 & -                 & -                 & 0.999999999999999\\\hline
\end{tabular}%
}\vspace{-0.2cm}
\end{table*}

\begin{table*}[]
\centering
\caption{Service reliability with $N=2$ different reliability classes in configuration $k=4$, $r=3$ and $\Psi=4$.}
\label{tabLast}\vspace{-0.3cm}
\resizebox{\textwidth}{!}{%
\begingroup
\setlength{\tabcolsep}{0.7pt} 
\renewcommand{\arraystretch}{1.2} 
\begin{tabular}{|r|c|c|c|c|c|c|c|c|c|}
\hline
active\textbackslash{}backip & $N_r$ & 4 & 4 & 4 & 4 & 3 & 3 & 3 & 3 \\ \cline{3-10} 
$N_r$ & $\Delta$ & 1 & 2 & 3 & 4 & 1 & 2 & 3 & 4 \\ \hline
4 & 1 & \textit{\textbf{0.9999979324510640}} & \textit{\textbf{0.9999979318655040}} & \textit{\textbf{0.9999979318286700}} & \textit{\textbf{0.9999979318251840}} & \textit{\textbf{0.9999964124393610}} & \textit{\textbf{0.9999964118583170}} & \textit{\textbf{0.9999964118218760}} & \textit{\textbf{0.9999961885402720}} \\
4 & 2 & \textit{\textbf{0.9999979317835310}} & \textit{\textbf{0.9999979317005130}} & \textit{\textbf{0.9999979316975760}} & \textit{\textbf{0.9999979316973230}} & \textit{\textbf{0.9999963106145000}} & \textit{\textbf{0.9999963105326350}} & \textit{\textbf{0.9999963105297770}} & \textit{\textbf{0.9999961883993840}} \\
4 & 3 & \textit{\textbf{0.9999979317535350}} & \textit{\textbf{0.9999979316968350}} & \textit{\textbf{0.9999979316951940}} & \textit{\textbf{0.9999979316950590}} & \textit{\textbf{0.9999963005396360}} & \textit{\textbf{0.9999963004838190}} & \textit{\textbf{0.9999963004822340}} & \textit{\textbf{0.9999961883957160}} \\
4 & 4 & \textit{\textbf{0.9999979317508330}} & \textit{\textbf{0.9999979316965460}} & \textit{\textbf{0.9999979316950160}} & \textit{\textbf{0.9999979316948910}} & \textit{\textbf{0.9999962995331520}} & \textit{\textbf{0.9999962994797230}} & \textit{\textbf{0.9999962994782460}} & \textit{\textbf{0.9999961883954060}} \\ \hline
3 & 1 & {{0.9999879366414630}} & {{0.9999879360639370}} & {{0.9999879360276900}} & {{0.9999879360242610}} & {{0.9999864169785840}} & {{0.9999864164055270}} & {{0.9999864163696700}} & {{0.9999861942812340}} \\
3 & 2 & 0.9999879359819710 & 0.9999879359019470 & 0.9999879358991620 & 0.9999879358989240 & 0.9999863151586640 & 0.9999863150797630 & 0.9999863150770560 & 0.9999861941428650 \\
3 & 3 & 0.9999879359525630 & 0.9999879358984220 & 0.9999879358968970 & 0.9999879358967720 & 0.9999863050840810 & 0.9999863050307960 & 0.9999863050293240 & 0.9999861941393090 \\
3 & 4 & 0.9999579366144800 & 0.9999579365603410 & 0.9999579365588160 & 0.9999579365586910 & 0.9999563057949240 & 0.9999563057416410 & 0.9999563057401690 & 0.9999561948534810 \\ \hline
2 & 1 & 0.9998879792477240 & 0.9998879787470810 & 0.9998879787164170 & 0.9998879787135260 & 0.9998864630738730 & 0.9998864625772640 & 0.9998864625469480 & 0.9998862523920450 \\
2 & 2 & 0.9998879786661850 & 0.9998879786134870 & 0.9998879786120130 & 0.9998879786118920 & 0.9998863613009290 & 0.9998863612490670 & 0.9998863612476450 & 0.9998862522773380 \\
2 & 3 & 0.9995880422682250 & 0.9995880422155420 & 0.9995880422140690 & 0.9995880422139490 & 0.9995864253881300 & 0.9995864253362840 & 0.9995864253348620 & 0.9995863163972430 \\
2 & 4 & 0.9995580549268320 & 0.9995580548741510 & 0.9995580548726770 & 0.9995580548725570 & 0.9995564380952430 & 0.9995564380433980 & 0.9995564380419760 & 0.9995563291076250 \\ \hline
1 & 1 & 0.9988884761855780 & 0.9988884761457560 & 0.9988884761447220 & 0.9988884761446380 & 0.9988869949822420 & 0.9988869949430760 & 0.9988869949420800 & 0.9988869043083370 \\
1 & 2 & 0.9958948064235610 & 0.9958948063838590 & 0.9958948063828270 & 0.9958948063827440 & 0.9958933296593930 & 0.9958933296203440 & 0.9958933296193510 & 0.9958932392572380 \\
1 & 3 & 0.9955960678574820 & 0.9955960678177920 & 0.9955960678167610 & 0.9955960678166780 & 0.9955945915362990 & 0.9955945914972620 & 0.9955945914962690 & 0.9955945011612620 \\
1 & 4 & 0.9955662002741250 & 0.9955662002344360 & 0.9955662002334050 & 0.9955662002333220 & 0.9955647239972310 & 0.9955647239581950 & 0.9955647239572020 & 0.9955646336249050 \\ \hline
\end{tabular}%
\endgroup
}\vspace{-0.6cm}
\end{table*}

\section{Numerical Results}\label{res}
We now evaluate the placement and component dependent SFC reliability considering all possible combinations of disjointedness level $\Delta$ with minimal number protected component types $N_r$ resulting in $16$ combinations following two placement strategies. The first strategy is " VNF placement based on SFC", where VNFs of different types are placed together, and the second is a "placement based on VNF type", where  VNFs of the same type are placed together. We consider two scenarios: 1) Case 1 relates to practice, where $p_1=0.99999$, $p_2=0.9999$, $p_3=0.999$, $p_4=0.99$; and 2) Case 2 is idealized scenario, where $p_1=p_2=p_3=p_4=0.99999$. We investigate the SFC reliability over $k=1, 4, 8$ active SFCs with $\Psi=4, 8$ VNFs and different amount of  backup SFCs $r$. We verified analytical results calculated with Eqs. \eqref{RDelta} and \eqref{RNPsiRoot} by Monte-Carlo simulations. Since the simulation results overlapped with analytical results (with 95\% confidence interval), we present only analytical results in a tabular form for clarity, where $N_r$ and in most cases SFC reliability values are sorted in the descending order. 

First, we assume that there is only one reliability class $N=1$ one active sub-SFC and evaluate the impact of $\Delta$, $N_r$, $\Psi$ and $r$ on the resulting SFC reliability.
Tab. \ref{compare} shows the SFC reliability for Case 1 and Case 2 as a function of the amount of backup sub-SFC $r$, SFC length $\Psi$ and different VNF placement strategies defined through $\Delta$ and $N_r$. The related placement (Pl.) is marked with $v$ for ``VNF based'' and with $s$ for ``SFC based'' strategies.
Without backup protection ($r=0$), the SFC reliability is independent from $N_r$, since there are no backup components. However, as expected, for a fixed  $N_r$  the SFC reliability decreases with increasing level of disjointedness $\Delta$ and increasing SFC length $\Psi$ (for fixed  $\Delta$ but different $N_r = 1,..,4$ the reliabilities are identical). As expected, Case 2 provides higher SFC reliability due to higher component reliability defined, but reflects the same functional dependency. Without any protection, a disjointedness of $\Delta=1$ will provide the highest end-to-end service reliability. 

Now, we consider the usual 1+1 backup protection, i.e., $r=1$.  At first glance, the SFC reliability decreases with decreasing amount of backup component types $N_r$, and an increasing level of disjointedness $\Delta$, e.g., for Case 1 with $\Psi=4$. Intuitive, that seems to be a practical rule. However, some placement strategies with low $N_r$ provide higher SFC reliability (marked bold) than some placements with a higher $N_r$. In Case 1 with a SFC length of $\Psi=8$,  placement with $N_r=3$ and $\Delta=1$ outperforms placements with $N_r=4$ and $\Delta=2,3,4$. That is because the SFC in these placement strategies will be most likely interrupted due to VNF failures caused by the longer VNF chain. Then, if we consider the VNF components of a SFC only, its reliability relates to $(p_4)^{\Psi}$, i.e., the lower reliability of the longer SFC chain, i.e., $(p_4)^{8} \approx 0.92274 < (p_4)^{4} \approx 0.9606 $, can not be compensated enough by the backup DC, rack and server. Thus, the amount of other backup (hardware) components, i.e., $N_r$, can be reduced. When we take a look on Case 2, where the placement strategies based on VNF type are marked bold. In this example we observe, that a placement  based on SFC with lower $N_r$ outperform a placement based on VNF type, e.g., placement with $N_r=2$ and $\Delta=1$ outperforms placement with $N_r=3$ and $\Delta=4$.  Moreover, also other placement strategies based on VNF type (marked bold) provide the lowest SFC reliability for any value of $r$ and $\Psi$ and it is preferable to use a placement strategy based on SFC. 

Next, let us consider the impact of increasing amount of backup sub-SFCs on resulting SFC reliability, where $r=100$ can be viewed as a good approximation of an upper bound. 
First, a more detailed rule can be observed for both use cases, which also reflects the case without protection: \textit{ For a given value of $N_r$, the reliability decreases with increasing  level of disjointedness $\Delta$, and, for a fixed value of  $\Delta$ the reliability increases with increasing $N_r$. }
For instance,  for a  $\Delta=1$ and descending $N_r=4...1$, as can be seen for all $r$ and $\Psi$, the higher $N_r$ results in the higher reliability.  Furthermore, the reliability increase with increasing redundancy, i.e., $~0.9$, $~0.999$, $~1$ with $r=0,1,100$, respectively, for $N_r=4$ and $\Delta=1$ in Case 1.
 In contrast, with $N_r=2$ and $N_r=1$ and $\Delta=1$ for both, the provided SFC reliability of ``3-'' and ``2-''nines is almost independent from value $r\neq0$, respectively.
As discussed above and verified by the results in Tables \ref{compare}, the highest SFC reliability in any configuration can be reached with $N_r=4$ and $\Delta=1$, where each SFC is allocated in the same server, but VNFs of the same type are placed in different DCs allowing protection of DC, rack, server and VNFs, i.e., $N_r=4$. 

Thus, next we consider only the placement with $N_r=4$ and $\Delta=1$, which provides the highest service reliability, in setting with $k\geq 1$ active sub-SFCs and Case 1. Tab. \ref{res} shows the service reliability as a function of $k$, $r$ and $\Psi$. The first column contain the amount of backup sub-SFCs utilized normalized by the number of active sub-SFCs. The increasing number of VNFs in SFC, i.e., $\Psi$, decreases the service reliability. In contrast, the increasing number of active sub-SFC significantly increases the service reliability if $r>0$. The reliability of ``six nines'' can be reached with configuration $k=4$ and $r=4$ or $k=8$ and $r=4$ for both $\Psi=4$ and $\Psi=8$, while serial SFC ($k=1$) requires $r=3$ backup sub-SFCs, which means $300\%$ redundancy.

Finally, we investigate the SFC reliability as a function of  placement strategies different for active and backup SFCs resulting in $N=2$ different reliability classes. We use configuration with $k=4$, $r=3$, $\Psi=4$ and component reliability defined for Case 1. Tab. \ref{tabLast} shows the resulting SFC reliability, where the first two columns and the two upper rows ($N_r$ and $\Delta$) describe placement strategy of active and backup sub-SFCs, respectively. Also here the general rule is for a fixed  $N_r$  the SFC reliability decreases with increasing level of disjointedness $\Delta$. Thus, we show the results for backup placement for $N_r=4, 3$ as that provides the highest SFC reliability. However, the highest reliability of \lq{}five nines\rq{} can be reached, when active sub-SFCs are placed with heterogeneity degree $N_r=4$ (marked bold italic).

\section{Conclusion}
In this paper, we provided a complete and generalized, simple reliability analysis, that considers system component sharing, heterogeneity and multitude of components as well as their interdependencies. Our analysis is based on combinatorial analysis and a reduced binomial theorem. Surprisingly, this simple approach could be utilized to analyzing rather complex and generic SFC configurations. The analytical results confirmed by simulation showed some trade-offs between component heterogeneity and disjointedness for SFC placement, whereby the increasing disjointedness level does not automatically lead to the highest SFC reliability.

%
%
%

 \bibliographystyle{IEEEtran}
 \bibliography{bibL}
\section{Appendix}

\subsection{Derivation of Eq. \eqref{LACI}}\label{firstLambda}
Since all components of a certain reliability class are placed separately from components of any other reliability class and only some common root components $c_{w_\rho}$ can be utilized to combine different classes, there is a need to consider each reliability class $\xi$ separately, whereby $\Lambda_{c_{\xi}}$ is a functions of VNF placement strategy and amount of available and failed components related to the reliability class $\xi$. Since DCs are the components from the highest hierarchy level and independent from failures of any other component types, the amount of available DCs is equal to the amount of data centers required to allocate $n_\xi$ sub-SFCs of reliability class $\xi$, i.e., $\Lambda_{1_\xi}=\bar n_{1_\xi}$. The failure of $f_{1_\xi}$ DCs causes a failure of $f_{1_\xi} n_{2_\xi}$ racks reducing the overall number of available racks as $\Lambda_{2_\xi}=(\Lambda_{1_\xi}-f_{1_\xi})n_{2_\xi}$, where $\Lambda_{2_\xi}\equiv \Lambda_{c_\xi}=(\Lambda_{c-1_\xi}-f_{c-1_\xi})n_{c_\xi}$. Then, the remaining amount of servers after DC and rack failures is $\Lambda_{3_\xi}=((\Lambda_{1_\xi}-f_{1_\xi})n_{2_\xi}-f_{2_\xi})n_{3_\xi}=(\Lambda_{2_\xi}-f_{2_\xi})n_{3_\xi}$, where $\Lambda_{3_\xi}\equiv\Lambda_{c_\xi}=(\Lambda_{c-1_\xi}-f_{c-1_\xi})n_{c_\xi}$. After failures of $f_{1_\xi}$ DCs, $f_{2_\xi}$ racks and $f_{3_\xi}$ servers the remaining amount of VNFs of any type is $\Lambda_{4_\xi}=(((\Lambda_{1_\xi}-f_{1_\xi})n_{2_\xi}-f_{2_\xi})n_{3_\xi}-f_{3_\xi})n_{4_\xi}=(\Lambda_{3_\xi}-f_{3_\xi})n_{4_\xi}\equiv (\Lambda_{c-1_\xi}-f_{c-1_\xi})n_{c_\xi}$, if $c_\xi=4$. 
Since some reliability classes have a common root component $c_{w_\rho}$, which impact any reliability class $\xi$, $\xi\in w_\rho$, there is a need to consider the availability and failure of common root components from a set $\Phi$ as $n_{c_\xi}\equiv n_{c_{w_\rho}}=\Lambda_{c_{w_{\rho}}}$ for any $\xi\in w_\rho$. Thus, the general expression for the amount of available components of any type $c$ from a reliability class $\xi$ is represented by Eq. \eqref{LACI}.


\end{document}